\documentclass[12pt]{article}
\usepackage[latin9]{inputenc}
\usepackage{geometry}
\usepackage{color}
\usepackage{amsmath}
\usepackage{amssymb}
\usepackage{stmaryrd}
\usepackage{graphicx}
\usepackage{setspace}
\usepackage[authoryear]{natbib}
\usepackage[unicode=true,
 bookmarks=false,
 breaklinks=false,pdfborder={0 0 1},backref=section,colorlinks=false]
 {hyperref}
\hypersetup{
 colorlinks,citecolor=blue,pdftex}

\makeatletter

\providecommand{\tabularnewline}{\\}

\usepackage{amsfonts}
\usepackage{amsthm}
\usepackage{lscape}
\usepackage{appendix}
\usepackage{dcolumn}
\usepackage{rotating}
\usepackage{comment}
\usepackage{color}

\usepackage{pgf}
\usepackage{tikz}\usepackage{caption}
\usepackage{subcaption}
\usepackage{tkz-tab}
\usepackage{tikz-3dplot}
\usetikzlibrary{calc}
\usetikzlibrary{arrows, automata}

\newcolumntype{d}[1]{D{.}{.}{#1}}
\newcolumntype{t}[1]{D{,}{,}{#1}}
\newcolumntype{i}[1]{D{.}{}{#1}}
\newtheorem{theorem}{Theorem}[section]

\newtheorem{definition}{Definition}[section]

\newtheorem{lemma}{Lemma}[section]

\newtheorem{proposition}{Proposition}[section]

\theoremstyle{plain}

\newtheorem*{asSX}{Assumption SX}

\newtheorem*{asL}{Assumption L}
\newtheorem*{asK}{Assumption K}

\newtheorem*{asM1}{Assumption M1}
\newtheorem*{asM2}{Assumption M2}

\newtheorem*{asB}{Assumption B}

\newtheorem*{asCS}{Assumption CS}

\numberwithin{equation}{section}

\makeatother

\begin{document}
\title{Optimal Dynamic Treatment Regimes\\
 and Partial Welfare Ordering\thanks{For helpful comments and discussions, the author is grateful to Donald
Andrews, Isaiah Andrews, Junehyuk Jung, Yuichi Kitamura, Hiro Kaido,
Shakeeb Khan, Adam McCloskey, Susan Murphy, Takuya Ura, Ed Vytlacil,
Shenshen Yang, participants in the 2021 Cowles Conference, the 2021
North American Winter Meeting, the 2020 European Winter Meeting, and
the 2020 World Congress of the Econometric Society, 2019 CEMMAP \&
WISE conference, the 2019 Asian Meeting of the Econometric Society,
the Bristol Econometrics Study Group Conference, the 2019 Midwest
Econometrics Group Conference, the 2019 Southern Economic Association
Conference, and in seminars at UCL, U of Toronto, Simon Fraser U,
Rice U, UIUC, U of Bristol, Queen Mary London, NUS, and SMU. \UrlFont{\sffamily}\protect\href{mailto:sukjin.han@gmail.com}{sukjin.han@gmail.com}}}
\author{Sukjin Han\\
 Department of Economics\\
 University of Bristol}
\date{First Draft: August 8, 2019 \\
 This Draft: \today}

\maketitle
\vspace{-0.2cm}

\begin{abstract}
Dynamic treatment regimes are treatment allocations tailored to heterogeneous
individuals. The optimal dynamic treatment regime is a regime that
maximizes counterfactual welfare. We introduce a framework in which
we can partially learn the optimal dynamic regime from observational
data, relaxing the sequential randomization assumption commonly employed
in the literature but instead using (binary) instrumental variables.
We propose the notion of sharp partial ordering of counterfactual
welfares with respect to dynamic regimes and establish mapping from
data to partial ordering via a set of linear programs. We then characterize
the identified set of the optimal regime as the set of maximal elements
associated with the partial ordering. We relate the notion of partial
ordering with a more conventional notion of partial identification
using topological sorts. Practically, topological sorts can be served
as a policy benchmark for a policymaker. We apply our method to understand
returns to schooling and post-school training as a sequence of treatments
by combining data from multiple sources. The framework of this paper
can be used beyond the current context, e.g., in establishing rankings
of multiple treatments or policies across different counterfactual
scenarios.

\vspace{0.1in}

\noindent \textit{JEL Numbers:} C14, C32, C33, C36

\noindent \textit{Keywords:} Optimal dynamic treatment regime, endogenous
treatments, dynamic treatment effect, partial identification, instrumental
variable, linear programming. 
\end{abstract}

\section{Introduction\label{sec:Introduction}}

Dynamic treatment regimes are dynamically personalized treatment
allocations. Given that individuals are heterogeneous, allocations
tailored to heterogeneity can improve overall welfare. Define a dynamic
treatment regime $\boldsymbol{\delta}(\cdot)$ as a sequence of binary
rules $\delta_{t}(\cdot)$ that map previous outcome and treatment
(and possibly other covariates) onto current allocation decisions:
$\delta_{t}(y_{t-1},d_{t-1})=d_{t}\in\{0,1\}$ for $t=1,...,T$.
The motivation for being adaptive to the previous outcome is that
it may contain information on unobserved heterogeneity that is not
captured in covariates. Then the optimal dynamic treatment regime,
which is this paper's main parameter of interest, is defined as a
regime that maximizes certain counterfactual welfare:
\begin{align}
\boldsymbol{\delta}^{*}(\cdot) & =\arg\max_{\boldsymbol{\delta}(\cdot)}W_{\boldsymbol{\delta}}.\label{eq:optDTR}
\end{align}
This paper investigates the possibility of identifiability of the
optimal dynamic regime $\boldsymbol{\delta}^{*}(\cdot)$ from data
that are generated from randomized experiments in the presence of
non-compliance or more generally from observational studies in multi-period
settings.

Optimal treatment regimes have been extensively studied in the biostatistics
literature (\citet{murphy2001marginal}, \citet{murphy2003optimal},
and \citet{robins2004optimal}, among others). These studies typically
rely on an ideal multi-stage experimental environment that satisfies
sequential randomization. Based on such experimental data, they identify
optimal regimes that maximize welfare, defined as the average counterfactual
outcome. However, non-compliance is prevalent in experiments, and
more generally, treatment endogeneity is a marked feature in observational
studies.\footnote{This point is also acknowledged as a concluding remark in \citet{murphy2001marginal}.}
This may be one reason the vast biostatistics literature has not yet
gained traction in other fields of social science, despite the potentially
fruitful applications of optimal dynamic regimes in various policy
evaluations.

To illustrate the policy relevance of the optimal dynamic regime,
consider the labor market returns to high school education and post-school
training for disadvantaged individuals. A policymaker may be interested
in learning a schedule of allocation rules $\boldsymbol{\delta}(\cdot)=(\delta_{1},\delta_{2}(\cdot))$
that maximizes the employment rate $W_{\boldsymbol{\delta}}=E[Y_{2}(\boldsymbol{\delta})]$,
where $\delta_{1}\in\{0,1\}$ assigns a high school diploma, $\delta_{2}(y_{1},\delta_{1})\in\{0,1\}$
assigns a job training program based on $\delta_{1}$ and earlier
earnings $y_{1}\in\{0,1\}$ (low or high), and $Y_{2}(\boldsymbol{\delta})$
indicates the counterfactual employment status under regime $\boldsymbol{\delta}(\cdot)$.
Suppose the optimal regime $\boldsymbol{\delta}^{*}(\cdot)$ is such
that $\delta_{1}^{*}=1$, $\delta_{2}^{*}(0,\delta_{1}^{*})=1$, and
$\delta_{2}^{*}(1,\delta_{1}^{*})=0$; that is, it turns out optimal
to assign a high school diploma to all individuals and a training
program to individuals with low earnings. One of the policy implications
of such $\boldsymbol{\delta}^{*}(\cdot)$ is that the average job
market performance can be improved by job trainings focusing on low
performance individuals complementing with high school education.
A static regime---where $\delta_{t}(\cdot)$ is a constant function---is
a special case of a dynamic regime. In this sense, the optimal dynamic
regime provides richer policy candidates than what can be learned
from dynamic complementarity (\citet{cunha2007technology}, \citet{cellini2010value},
\citet{almond2013fetal}, \citet{johnson2019reducing}). In learning
$\boldsymbol{\delta}^{*}(\cdot)$ in this example, observational data
may only be available where the observed treatments (schooling decisions)
are endogenous.

This paper proposes a nonparametric framework, in which we can at
least partially learn the ranking of counterfactual welfares $W_{\boldsymbol{\delta}}$'s
and hence the optimal dynamic regime $\boldsymbol{\delta}^{*}(\cdot)$.
We view that it is important to avoid making stringent modeling assumptions
in the analysis of personalized treatments, because the core motivation
of the analysis is individual heterogeneity, which we want to keep
intact as much as possible. Instead, we embrace the partial identification
approach. Given the observed distribution of sequences of outcomes
and endogenous treatments and using the instrumental variable (IV)
method, we establish sharp partial ordering of welfares, and characterize
the identified set of optimal regimes as a discrete subset of all
possible regimes. We define welfare as a linear functional of the
joint distribution of counterfactual outcomes across periods. Examples
of welfare include the average counterfactual terminal outcome commonly
considered in the literature and as shown above.  We assume we are
equipped with some IVs that are possibly binary. We show that it is
helpful to have a sequence of IVs generated from sequential experiments
or quasi-experiments. Examples of the former are increasingly common
as forms of random assignments or encouragements in medical trials,
public health and educational interventions, and A/B testing on digital
platforms. Examples of the latter can be some combinations of traditional
IVs and regression discontinuity designs. Our framework also accommodates
a single binary IV in the context of dynamic treatments and outcomes
(e.g., \citet{cellini2010value}). The identifying power in such a
case is investigated in simulation. The partial ordering and identified
set proposed in this paper enable ``sensitivity analyses.'' That
is, by comparing a chosen regime (e.g., from a parametric approach)
with these benchmark objects, one can determine how much the former
is led by assumptions and how much is informed by data. Such a practice
also allows us to gain insight into data requirements to achieve a
certain level of informativeness.

\begin{figure}
\begin{centering}
\includegraphics[scale=0.28]{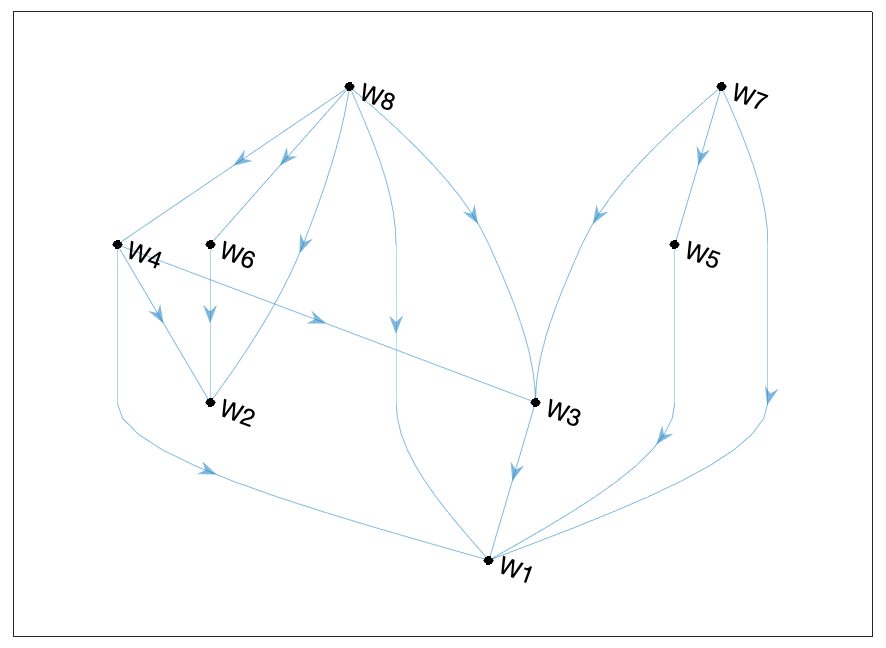}
\par\end{centering}
\caption{An Example of Sharp Partial Ordering of Welfares}
\label{fig:dag-1}
\end{figure}
The identification analysis is twofold. In the first part, we establish
mapping from data to sharp partial ordering of counterfactual welfares
with respect to possible regimes. The point identification of $\boldsymbol{\delta}^{*}(\cdot)$
will be achieved by establishing the total ordering of welfares, which
is not generally possible in this flexible nonparametric framework
with limited exogenous variation. Figure \ref{fig:dag-1} is an example
of partial ordering that we calculated by applying this paper's theory
and using simulated data. Here, we consider a two-period case as in
the motivating example above, which yields eight possible $\boldsymbol{\delta}(\cdot)$'s
and corresponding welfares, and ``$\rightarrow$'' corresponds to
the relation ``$>$''. To establish the partial ordering, we first
characterize bounds on the difference between a pair of welfares as
the set of optima of linear programs, and we do so for all possible
welfare pairs. The bounds on welfare gaps are informative about whether
welfares are comparable or not, and when they are, how to rank them.
Then we show that although the bounds are calculated from separate
optimizations, the partial ordering is consistent with \textit{common}
data-generating processes. The partial ordering obtained in this
way is shown to be sharp in the sense that will become clear later.
Note that each welfare gap measures the \textit{dynamic treatment
effect}. The partial ordering concisely (and tightly) summarizes
the identified signs of these treatment effects, and thus can be a
parameter of independent interest.

In the second part of the analysis, given the sharp partial ordering,
we show that the identified set can be characterized as the set of
maximal elements associated with the partial ordering, i.e., the set
of regimes that are \textit{not inferior}. For example, according
to Figure \ref{fig:dag-1}, the identified set consists of regimes
7 and 8. Given the partial ordering, we also calculate topological
sorts, which are total orderings that do not violate the underlying
partial ordering. Theoretically, topological sorts can be viewed as
observationally equivalent total orderings, which insight relates
the partial ordering we consider with a more conventional notion of
partial identification. Practically, topological sorts can be served
as a policy benchmark that a policymaker can be equipped with. If
desired, linear programming can be solved to calculate bounds on a
small number of sorted welfares (e.g., top-tier welfares).

Given the minimal structure we impose in the data-generating process,
the size of the identified set may be large in some cases. Such an
identified set may still be useful in eliminating suboptimal regimes
or warning about the lack of informativeness of the data. Often,
however, researchers are willing to impose additional assumptions
to gain identifying power. We propose identifying assumptions, such
as uniformity assumptions that generalize the monotonicity assumption
in \citet{imbens1994identification}, an assumption about an agent's
learning, Markovian structure, and stationarity. These assumptions
tighten the identified set by reducing the dimension of the simplex
in the linear programming, thus producing a denser partial ordering.
We show that these assumptions are easy to impose in our framework.

This paper makes several contributions. To our best knowledge, this
paper is first in the literature that considers the identifiability
of optimal \textit{dynamic} regimes under treatment endogeneity. The
pioneering work by \citet{murphy2003optimal} and subsequent works
consider point identification of optimal dynamic regimes under the
sequential randomization assumption. This paper brings this literature
to observational contexts. Recently, \citet{han2018nonparametric},
\citet{han2020comment}, \citet{cui2019semiparametric}, and \citet{qiu2020optimal}
relax sequential randomization and establish identification of dynamic
average treatment effects and/or optimal regimes using instrumental
variables. In a single-period setup, they consider a regime that is
a mapping only from covariates, but not previous outcomes and treatments,
to an allocation. They focus on point identification by imposing assumptions
such as the existence of additional exogenous variables in a multi-period
setup (\citet{han2018nonparametric}), or the zero correlation between
unmeasured confounders and compliance types (\citet{cui2019semiparametric,qiu2020optimal})
or uniformity (\citet{han2020comment}). The dynamic effects of treatment
timing (i.e., irreversible treatments) have been considered in \citet{heckman2007dynamic}
and \citet{heckman2016dynamic} who utilize exclusion restrictions
and infinite support assumptions. A related staggered adoption design
was recently studied in multi-period difference-in-differences settings
under treatment heterogeneity by \citet{athey2018design}, \citet{callaway2018difference},
and \citet{abraham2018estimating}.\footnote{\citet{de2020two} consider a similar problem but without necessarily
assuming staggered adoption.} This paper complements these papers by considering treatment scenarios
of multiple dimensions with adaptivity as the key ingredient.

Second, this paper contributes to the literature on partial identification
that utilizes linear programming approach, which has early examples
as \citet{balke1997bounds} and \citet{manski2007partial}, and appears
recently in \citet{torgovitsky2016partial}, \citet{deb2017revealed},
\citet{mogstad2018using}, \citet{kitamura2019nonparametric}, \citet{machado2018instrumental},
\citet{tebaldi2019nonparametric}, \citet{kamat2017identification},
\citet{gunsilius2019bounds}, and \citet{han2020sharp}, to name a
few. The advantages of this approach is that (i) bounds can be automatically
obtained even when analytical derivation is not possible, (ii) the
proof of sharpness is straightforward and not case-by-case, and (iii)
it can streamline the analysis of different identifying assumptions.
The dynamic framework of this paper complicates the identification
analysis, which therefore fully benefits from these advantages. However,
a distinct feature of the present paper is that the linear programming
approach is used in establishing a sharp partial ordering across counterfactual
objects---a novel concept in the literature---and in such a way
that separate optimizations yield a common object, namely the partial
ordering. The framework of this paper can also be useful in other
settings where the goal is to compare welfares across multiple treatments
and regimes---e.g., personalized treatment rules---or more generally,
to establish rankings of policies across different counterfactual
scenarios and find the best ones.

Third, we apply our method to conduct a policy analysis with schooling
and post-school training as a sequence of treatments, which is to
our knowledge a novel attempt in the literature. We consider dynamic
treatment regimes of allocating a high school diploma and, given pre-program
earnings, a job training program for economically disadvantaged population.
By combining data from the Job Training Partnership Act (JTPA), the
US Census, and the National Center for Education Statistics (NCES),
we construct a data set with a sequence of instruments that is used
to estimate the partial ordering of expected earnings and the identified
set of the optimal regime. Even though only partial orderings are
recovered, we can conclude with certainty that allocating the job
training program only to the low earning type is welfare optimal.
We also find that more costly regimes are not necessary welfare-improving.

The dynamic treatment regime considered in this paper is broadly related
to the literature on statistical treatment rules, e.g., \citet{manski2004statistical},
\citet{hirano2009asymptotics}, \citet{bhattacharya2012inferring},
\citet{stoye2012minimax}, \citet{kitagawa2018should}, \citet{kasy2016partial},
and \citet{athey2017efficient}. However, our setting, assumptions,
and goals are different from those in these papers. In a single-period
setting, they consider allocation rules that map covariates to decisions.
They impose assumptions that ensure point identification, such as
(conditional) unconfoundedness, and focus on establishing the asymptotic
optimality of the treatment rules, with \citet{kasy2016partial} the
exception.\footnote{\citet{athey2017efficient}'s framework allows observational data
with endogenous treatments as a special case, but the conditional
homogeneity of treatment effects is assumed.} \citet{kasy2016partial} focuses on establishing partial ranking
by comparing a pair of treatment-allocating probabilities as policies.
The notion of partial identification of ranking is related to ours,
but we introduce the notion of sharpness of a partially ordered set
with discrete policies and a linear programming approach to achieve
that. Another distinction is that we consider a dynamic setup. Finally,
in order to focus on the challenge with endogeneity, we consider a
simple setup where the exploration and exploitation stages are separated,
unlike in the literature on bandit problems (\citet{kock2017optimal},
\citet{kasy2019adaptive}, \citet{athey2019machine}). We believe
the current setup is a good starting point.

In the next section, we introduce the dynamic regimes and related
counterfactual outcomes, which define the welfare and the optimal
regime. Section \ref{sec:Motivating-Example:-Returns} provides a
motivating example. Section \ref{sec:Partial-Ordering-and} conducts
the main identification analysis by constructing the partial ordering
and characterizing the identified set. Sections \ref{sec:Topological-Sorting}--\ref{sec:Cardinality-Reduction}
introduce topological sorts and additional identifying assumptions
and discuss cardinality reduction for the set of regimes. Section
\ref{sec:Numerical-Studies} illustrates the analysis with numerical
exercises, and Section \ref{sec:Application} presents the empirical
application on returns to schooling and job training. Finally, Section
\ref{sec:Discussions:-Estimation-and} concludes by discussing inference.
Most proofs are collected in the Appendix.

In terms of notation, let $\boldsymbol{W}^{t}\equiv(W_{1},..,W_{t})$
denote a vector that collects r.v.'s $W_{t}$ across time up to $t$,
and let $\boldsymbol{w}^{t}$ be its realization. Most of the time,
we write $\boldsymbol{W}\equiv\boldsymbol{W}^{T}$ for convenience.
We abbreviate ``with probability one'' as ``w.p.1'' and ``with
respect to'' as ``w.r.t.'' The symbol ``$\perp$'' denotes statistical
independence.

\section{Dynamic Regimes and Counterfactual Welfares\label{sec:Dynamic-Regimes-and}}

\subsection{Dynamic Regimes}

Let $t$ be the index for a period or stage. For each $t=1,...,T$
with fixed $T$, define an \textit{adaptive treatment rule} $\delta_{t}:\{0,1\}^{t-1}\times\{0,1\}^{t-1}\rightarrow\{0,1\}$
that maps the lags of the realized binary outcomes and treatments
$\boldsymbol{y}^{t-1}\equiv(y_{1},...,y_{t-1})$ and $\boldsymbol{d}^{t-1}\equiv(d_{1},...,d_{t-1})$
onto a deterministic treatment allocation $d_{t}\in\{0,1\}$:
\begin{align}
\delta_{t}(\boldsymbol{y}^{t-1},\boldsymbol{d}^{t-1}) & =d_{t}.\label{eq:trt_rule}
\end{align}
This adaptive rule also appears in, e.g., \citet{murphy2003optimal}.
The rule can also be a function of other discrete covariates, which
is a straightforward extension and thus is not considered here for
brevity. A special case of \eqref{eq:trt_rule} is a static rule where
$\delta_{t}(\cdot)$ is only a function of covariates but not $(\boldsymbol{y}^{t-1},\boldsymbol{d}^{t-1})$
(\citet{han2018nonparametric}, \citet{cui2019semiparametric}) or
a constant function.\footnote{This means that our term of ``static regime'' is narrowly defined
than in the literature. In the literature, a regime is sometimes called
dynamic even if it is only a function of covariates.} Binary outcomes and treatments are prevalent, and they are helpful
in analyzing, interpreting, and implementing dynamic regimes (\citet{zhang2015using}).
Still, extending the framework to allow for multi-valued discrete
variables is possible. Whether the rule is dynamic or static, we only
consider deterministic rules $\delta_{t}(\cdot)\in\{0,1\}$. In Appendix
\ref{subsec:Stochastic}, we extend this to stochastic rules $\tilde{\delta}_{t}(\cdot)\in[0,1]$
and show why it is enough to consider deterministic rules in some
cases. Then, a \textit{dynamic regime} up to period $t$ is defined
as a vector of all treatment rules: 
\begin{align*}
\boldsymbol{\delta}^{t}(\cdot) & \equiv\left(\delta_{1},\delta_{2}(\cdot),...,\delta_{t}(\cdot)\right).
\end{align*}
Let $\boldsymbol{\delta}(\cdot)\equiv\boldsymbol{\delta}^{T}(\cdot)\in\mathcal{D}$
where $\mathcal{D}$ is the set of all possible regimes.\footnote{We can allow $\mathcal{D}$ to be a strict subset of the set of all
possible regimes; see Section \ref{sec:Cardinality-Reduction} for
this relaxation.} Throughout the paper, we will mostly focus on the leading case with
$T=2$ for simplicity. Also, this case already captures the essence
of the dynamic features, such as adaptivity and complementarity. Table
\ref{tab:regimes} lists all possible dynamic regimes $\boldsymbol{\delta}(\cdot)\equiv\left(\delta_{1},\delta_{2}(\cdot)\right)$
as contingency plans. 
\begin{table}
\begin{centering}
\begin{tabular}{|c|c|c|c|}
\hline 
Regime \#  & $\delta_{1}$  & $\delta_{2}(1,\delta_{1})$  & $\delta_{2}(0,\delta_{1})$\tabularnewline
\hline 
\hline 
1  & 0  & 0  & 0\tabularnewline
\hline 
2  & 1  & 0  & 0\tabularnewline
\hline 
3  & 0  & 1  & 0\tabularnewline
\hline 
4  & 1  & 1  & 0\tabularnewline
\hline 
5  & 0  & 0  & 1\tabularnewline
\hline 
6  & 1  & 0  & 1\tabularnewline
\hline 
7  & 0  & 1  & 1\tabularnewline
\hline 
8  & 1  & 1  & 1\tabularnewline
\hline 
\end{tabular}
\par\end{centering}
\caption{Dynamic Regimes $\boldsymbol{\delta}(\cdot)$ When $T=2$}
\label{tab:regimes}
\end{table}

\subsection{Counterfactual Welfares and Optimal Regimes}

To define welfare w.r.t. this dynamic regime, we first introduce a
counterfactual outcome as a function of a dynamic regime. Because
of the adaptivity intrinsic in dynamic regimes, expressing counterfactual
outcomes is more involved than that with static regimes $d_{t}$,
i.e., $Y_{t}(\boldsymbol{d}^{t})$ with $\boldsymbol{d}^{t}\equiv(d_{1},...,d_{t})$.
Let $\boldsymbol{Y}^{t}(\boldsymbol{d}^{t})\equiv(Y_{1}(d_{1}),Y_{2}(\boldsymbol{d}^{2}),...,Y_{t}(\boldsymbol{d}^{t}))$.
We express a counterfactual outcome with adaptive regime $\boldsymbol{\delta}^{t}(\cdot)$
as follows\footnote{As the notation suggests, we implicitly assume the ``no anticipation''
condition.}: 
\begin{align}
Y_{t}(\boldsymbol{\delta}^{t}(\cdot)) & \equiv Y_{t}(\boldsymbol{d}^{t}),\label{eq:Y(delta)}
\end{align}
where the ``bridge variables'' $\boldsymbol{d}^{t}\equiv(d_{1},...,d_{t})$
satisfy
\begin{align}
d_{1} & =\delta_{1},\nonumber \\
d_{2} & =\delta_{2}(Y_{1}(d_{1}),d_{1}),\nonumber \\
d_{3} & =\delta_{3}(\boldsymbol{Y}^{2}(\boldsymbol{d}^{2}),\boldsymbol{d}^{2}),\label{eq:a's}\\
 & \vdots\nonumber \\
d_{t} & =\delta_{t}(\boldsymbol{Y}^{t-1}(\boldsymbol{d}^{t-1}),\boldsymbol{d}^{t-1}).\nonumber 
\end{align}
Suppose $T=2$. Then, the two counterfactual outcomes are defined
as $Y_{1}(\delta_{1})=Y_{1}(d_{1})$ and $Y_{2}(\boldsymbol{\delta}^{2}(\cdot))=Y_{2}(\delta_{1},\delta_{2}(Y_{1}(\delta_{1}),\delta_{1}))$.

Let $q_{\boldsymbol{\delta}}(\boldsymbol{y})\equiv\Pr[\boldsymbol{Y}(\boldsymbol{\delta}(\cdot))=\boldsymbol{y}]$
be the joint distribution of counterfactual outcome vector $\boldsymbol{Y}(\boldsymbol{\delta}(\cdot))\equiv(Y_{1}(\delta_{1}),Y_{2}(\boldsymbol{\delta}^{2}(\cdot)),...,Y_{T}(\boldsymbol{\delta}(\cdot)))$.
We define counterfactual welfare as a linear functional of $q_{\boldsymbol{\delta}}(\boldsymbol{y})$:
\begin{align*}
W_{\boldsymbol{\delta}} & \equiv f(q_{\boldsymbol{\delta}}).
\end{align*}
Examples of the functional include the average counterfactual terminal
outcome $E[Y_{T}(\boldsymbol{\delta}(\cdot))]=\Pr[Y_{T}(\boldsymbol{\delta}(\cdot))=1]$,
our leading case and which is common in the literature, and the weighted
average of counterfactuals $\sum_{t=1}^{T}\omega_{t}E[Y_{t}(\boldsymbol{\delta}^{t}(\cdot))]$.
Then, the \textit{optimal dynamic regime} is a regime that maximizes
the welfare as defined in \eqref{eq:optDTR}:\footnote{We assume that the optimal dynamic regime is unique by simply ruling
out knife-edge cases in which two regimes deliver the same welfare.}
\begin{align*}
\boldsymbol{\delta}^{*}(\cdot) & =\arg\max_{\boldsymbol{\delta}(\cdot)\in\mathcal{D}}W_{\boldsymbol{\delta}}.
\end{align*}
In the case of $W_{\boldsymbol{\delta}}=E[Y_{T}(\boldsymbol{\delta}(\cdot))]$,
the solution $\boldsymbol{\delta}^{*}(\cdot)$ can be justified by
backward induction in finite-horizon dynamic programming. Moreover
in this case, the regime with deterministic rules $\delta_{t}(\cdot)\in\{0,1\}$
achieves the same optimal regime and optimized welfare as the regime
with stochastic rules $\delta_{t}(\cdot)\in[0,1]$; see Theorem \ref{thm:stoch_regime}
in Appendix \ref{subsec:Stochastic}.

The identification analysis of the optimal regime is closely related
to the identification of welfare for each regime and welfare gaps,
which also contain information for policy. Some interesting special
cases are the following: (i) the \textit{optimal welfare}, $W_{\boldsymbol{\delta}^{*}}$,
which in turn yields (ii) the \textit{regret} from following individual
decisions, $W_{\boldsymbol{\delta}^{*}}-W_{\boldsymbol{D}}$, where
$W_{\boldsymbol{D}}$ is simply $f(\Pr[\boldsymbol{Y}(\boldsymbol{D})=\cdot])=f(\Pr[\boldsymbol{Y}=\cdot])$,
and (iii) the \textit{gain from adaptivity}, $W_{\boldsymbol{\delta}^{*}}-W_{\boldsymbol{d}^{*}}$,
where $W_{\boldsymbol{d}^{*}}=\max_{\boldsymbol{d}}W_{\boldsymbol{d}}$
is the optimum of the welfare with a static rule, $W_{\boldsymbol{d}}=f(\Pr[\boldsymbol{Y}(\boldsymbol{d})=\cdot])$.
If the cost of treatments is not considered, the gain in (iii) is
non-negative as the set of all $\boldsymbol{d}$ is a subset of $\mathcal{D}$.

\section{Motivating Examples\label{sec:Motivating-Example:-Returns}}

For illustration, we continue discussing the example in the Introduction.
This stylized example in an observational setting is meant to motivate
the policy relevance of the optimal dynamic regime and the type of
data that are useful for recovering it. Again, consider labor market
returns to high school education and post-school training for disadvantaged
individuals. Let $D_{i1}=1$ if student $i$ has a high school diploma
and $D_{i1}=0$ otherwise; let $D_{i2}=1$ if $i$ participates in
a job training program and $D_{i2}=0$ if not. Also, let $Y_{i1}=1$
if $i$ is employed before the training program and $Y_{i1}=0$ if
not; let $Y_{i2}=1$ if $i$ is employed after the program and $Y_{i2}=0$
if not. Given the data, suppose we are interested in recovering regimes
that maximize the employment rate as welfare.

First, consider a static regime, which is a schedule $\boldsymbol{d}=(d_{1},d_{2})\in\{0,1\}^{2}$
of first assigning a high school diploma ($d_{1}\in\{0,1\}$) and
then a job training ($d_{2}\in\{0,1\}$). Define associated welfare,
which is the employment rate $W_{\boldsymbol{d}}=E[Y_{2}(\boldsymbol{d})]$.
This setup is already useful in learning, for example, $E[Y_{2}(1,0)]-E[Y_{2}(0,1)]$
or complementarity (i.e., $E[Y_{2}(0,1)]-E[Y_{2}(0,0)]$ versus $E[Y_{2}(1,1)]-E[Y_{2}(1,0)]$),
which cannot be learned from period-specific treatment effects. However,
because $d_{1}$ and $d_{2}$ are not simultaneously given but $d_{1}$
precedes $d_{2}$, the allocation $d_{2}$ can be more informed by
incorporating the knowledge about the individual's response $y_{1}$
to $d_{1}$. This motivates the dynamic regime, which is the schedule
$\boldsymbol{\delta}(\cdot)=(\delta_{1},\delta_{2}(\cdot))\in\mathcal{D}$
of allocation \textit{rules} that first assigns a high school diploma
($\delta_{1}\in\{0,1\}$) and then a job training ($\delta_{2}(y_{1},\delta_{1})\in\{0,1\}$)
depending on $\delta_{1}$ and the employment status $y_{1}$. Then,
the optimal regime with adaptivity $\boldsymbol{\delta}^{*}(\cdot)$
is the one that maximizes $W_{\boldsymbol{\delta}}=E[Y_{2}(\boldsymbol{\delta})]$.
As argued in the Introduction, $\boldsymbol{\delta}^{*}(\cdot)$ provides
policy implications that $\boldsymbol{d}^{*}$ cannot.

As $D_{1}$ and $D_{2}$ are endogenous, $\{D_{i1},Y_{i1},D_{i2},Y_{i2}\}$
above are not useful by themselves to identify $W_{\boldsymbol{\delta}}$'s
and $\boldsymbol{\delta}^{*}(\cdot)$. We employ the approach of using
IVs, either a single IV (e.g., in the initial period) or a sequence
of IVs. In experimental settings, examples of a sequence of IVs can
be found in multi-stage experiments, such as the Fast Track Prevention
Program (\citet*{conduct1992developmental}), the Elderly Program
randomized trial for the Systolic Hypertension (\citet{the1988rationale}),
and Promotion of Breastfeeding Intervention Trial (\citet{kramer2001promotion}).
It is also possible to combine multiple experiments as in \citet{johnson2019reducing}.
In observational settings, one can use IVs from quasi-experiments,
those from RD design, or a combination of them. In the example above,
we can use the distance to high schools or the number of high schools
per square mile as an instrument $Z_{1}$ for $D_{1}$. Then, a random
assignment of the job training in a field experiment can be used as
an instrument $Z_{2}$ for the compliance decision $D_{2}$. In fact,
in Section \ref{sec:Application}, we study schooling and job training
as a sequence of treatments and combine IVs from experimental and
observational data.

\section{Partial Ordering and Partial Identification\label{sec:Partial-Ordering-and}}

\subsection{Observables}

We introduce observables based on which we want to identify the optimal
regime and counterfactual welfares. Assume that the time length of
the observables is equal to $T$, the length of the optimal regime
to be identified.\footnote{In general, we may allow $\tilde{T}\ge T$ where $\tilde{T}$ is the
length of the observables.} For each period or stage $t=1,...,T$, assume that we observe the
binary instrument $Z_{t}$, the binary endogenous treatment decision
$D_{t}$, and the binary outcome $Y_{t}=\sum_{\boldsymbol{d}^{t}\in\{0,1\}^{t}}1\{\boldsymbol{D}^{t}=\boldsymbol{d}^{t}\}Y_{t}(\boldsymbol{d}^{t})$.
These variables are motivated in the previous section. As another
example, $Y_{t}$ is a symptom indicator for a patient, $D_{t}$ is
the medical treatment received, and $Z_{t}$ is generated by a multi-period
medical trial. Importantly, the framework does not preclude the case
in which $Z_{t}$ exists only for some $t$ but not all; see Section
\ref{sec:Numerical-Studies} for related discussions. In this case,
$Z_{t}$ for the other periods is understood to be degenerate. Let
$D_{t}(\boldsymbol{z}^{t})$ be the counterfactual treatment given
$\boldsymbol{z}^{t}\equiv(z_{1},...,z_{t})\in\{0,1\}^{t}$. Then,
$D_{t}=\sum_{\boldsymbol{z}^{t}\in\mathcal{Z}^{t}}D_{t}(\boldsymbol{z}^{t})$.
Let $\boldsymbol{Y}(\boldsymbol{d})\equiv(Y_{1}(d_{1}),Y_{2}(\boldsymbol{d}^{2}),...,Y_{T}(\boldsymbol{d}))$
and $\boldsymbol{D}(\boldsymbol{z})\equiv(D_{1}(z_{1}),D_{2}(\boldsymbol{z}^{2}),...,D_{T}(\boldsymbol{z}))$.

\begin{asSX}$Z_{t}\perp(\boldsymbol{Y}(\boldsymbol{d}),\boldsymbol{D}(\boldsymbol{z}))|\boldsymbol{Z}^{t-1}$.\end{asSX}

Assumption SX assumes the strict exogeneity and exclusion restriction.\footnote{There may be other covariates available for the researcher, but we
suppress them for brevity. All the stated assumptions and the analyses
of this paper can be followed conditional on the covariates. A sufficient
condition for Assumption SX is that $\boldsymbol{Z}\perp(\boldsymbol{Y}(\boldsymbol{d}),\boldsymbol{D}(\boldsymbol{z}))$.} A single IV with full independence trivially satisfies this assumption.
For a sequence of IVs, this assumption is satisfied in typical sequential
randomized experiments, as well as quasi-experiments as discussed
in Section \ref{sec:Motivating-Example:-Returns}. Let $(\boldsymbol{Y},\boldsymbol{D},\boldsymbol{Z})$
be the vector of observables $(Y_{t},D_{t},Z_{t})$ for the entire
$T$ periods and let $p$ be its distribution. We assume that $(\boldsymbol{Y}_{i},\boldsymbol{D}_{i},\boldsymbol{Z}_{i})$
is independent and identically distributed and $\{(\boldsymbol{Y}_{i},\boldsymbol{D}_{i},\boldsymbol{Z}_{i}):i=1,...,N\}$
is a small $T$ large $N$ panel. We mostly suppress the individual
unit $i$ throughout the paper. For empirical applications, the data
structure can be more general than a panel and the kinds of $Y_{t}$,
$D_{t}$ and $Z_{t}$ are allowed to be different across time; Section
\ref{sec:Motivating-Example:-Returns} contains such an example. For
the population from which the data are drawn, we are interested in
learning the optimal regime.

\subsection{Partial Ordering of Welfares}

Given the distribution $p$ of the data $(\boldsymbol{Y},\boldsymbol{D},\boldsymbol{Z})$
and under Assumption SX, we show how the optimal dynamic regime and
welfares can be partially recovered. The identified set of $\boldsymbol{\delta}^{*}(\cdot)$
will be characterized as a subset of the discrete set $\mathcal{D}$.
As the first step, we establish \textit{partial ordering} of $W_{\boldsymbol{\delta}}$
w.r.t. $\boldsymbol{\delta}(\cdot)\in\mathcal{D}$ as a function of
$p$. The partial ordering summarizes the identified signs of the
dynamic treatment effects, as will become clear later. The partial
ordering can be represented by a \textit{directed acyclic graph} (DAG).\footnote{The way directed graphs are used in this paper is completely unrelated
to causal graphical models in the literature.} The DAG representation is fruitful for introducing the notion of
the sharpness of partial ordering and later to translate it into the
identified set of $\boldsymbol{\delta}^{*}(\cdot)$.

To facilitate this analysis, we enumerate all $\left|\mathcal{D}\right|=2^{2^{T}-1}$
possible regimes. For index $k\in\mathcal{K}\equiv\{k:1\le k\le\left|\mathcal{D}\right|\}$
(and thus $\left|\mathcal{K}\right|=\left|\mathcal{D}\right|$), let
$\boldsymbol{\delta}_{k}(\cdot)$ denote the $k$-th regime in $\mathcal{D}$.
For $T=2$, Table \ref{tab:regimes} indexes all possible dynamic
regimes $\boldsymbol{\delta}(\cdot)\equiv\left(\delta_{1},\delta_{2}(\cdot)\right)$.
Let $W_{k}\equiv W_{\boldsymbol{\delta}_{k}}$ be the corresponding
welfare. Figure \ref{fig:partial_order} illustrates examples of the
partially ordered set of welfares where each edge ``$W_{k}\rightarrow W_{k'}$''
indicates the relation ``$W_{k}>W_{k'}$.''
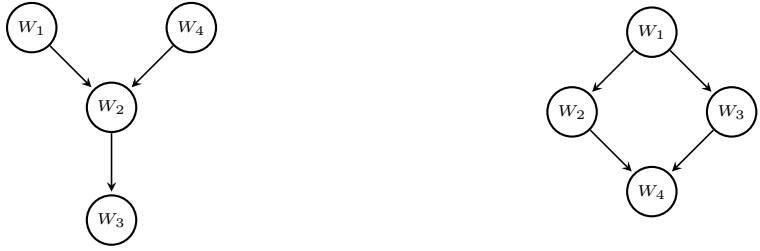
\begin{figure}
\centering

\begin{subfigure}[t]{0.35\textwidth} \centering \begin{tiny}
\begin{tikzpicture}[
> = stealth, shorten > = 1pt, auto,         
node distance = 1.5cm, semithick ]
      
\tikzstyle{every state}=[         
draw = black,         
thick,         
fill = white,         
minimum size = 3mm       
]
      
\node[state] (w2) { $W_2$ };       
\node[state] (w1) [above left of=w2] { $W_1$ };       
\node[state] (w4) [above right of=w2] { $W_4$ };       
\node[state] (w3) [below of=w2] { $W_3$ };       
      
\path[->] (w1) edge node {} (w2);       
\path[->] (w4) edge node {} (w2);       
\path[->] (w2) edge node {} (w3);     

\draw[white, dashed] (0, -2) -- (0, -2);

\end{tikzpicture} \end{tiny} \caption{$\boldsymbol{\delta}^{*}(\cdot)$ is partially identified}
\end{subfigure} ~~~~~~~~ \begin{subfigure}[t]{0.35\textwidth}
\centering \begin{tiny} \begin{tikzpicture}[
> = stealth, shorten > = 1pt, auto,         
node distance = 1.5cm, semithick ]
      
\tikzstyle{every state}=[         
draw = black,         
thick,         
fill = white,         
minimum size = 3mm       
]
      
\node[state] (w1) { $W_1$ };       
\node[state] (w2) [below left of=w1] { $W_2$ };       
\node[state] (w3) [below right of=w1] { $W_3$ };       
\node[state] (w4) [below right of=w2] { $W_4$ };       
      
\path[->] (w1) edge node {} (w2);       
\path[->] (w1) edge node {} (w3);       
\path[->] (w2) edge node {} (w4);       
\path[->] (w3) edge node {} (w4);     

\draw[white, dashed] (0, -3) -- (0, -3);

\end{tikzpicture} \end{tiny} \caption{$\boldsymbol{\delta}^{*}(\cdot)$ is point identified}
\end{subfigure}

\caption{Partially Ordered Sets of Welfares}
\label{fig:partial_order}
\end{figure}

In general, the point identification of $\boldsymbol{\delta}^{*}(\cdot)$
is achieved by establishing the total ordering of $W_{k}$, which
is not possible with instruments of limited support. Instead, we only
recover a partial ordering. We want the partial ordering to be sharp
in the sense that it cannot be improved given the data and maintained
assumptions. To formally state this, let $G(\mathcal{K},\mathcal{E})$
be a DAG where $\mathcal{K}$ is the set of welfare (or regime) indices
and $\mathcal{E}$ is the set of edges.

\begin{definition}\label{def:sharp_DAG}Given the data distribution
$p$, a partial ordering $G(\mathcal{K},\mathcal{E}_{p})$ is sharp
under the maintained assumptions if there exists no partial ordering
$G(\mathcal{K},\mathcal{E}_{p}')$ such that $\mathcal{E}_{p}'\supsetneq\mathcal{E}_{p}$
without imposing additional assumptions.

\end{definition}

Establishing sharp partial ordering amounts to determining whether
we can tightly identify the sign of a counterfactual welfare gap $W_{k}-W_{k'}$
(i.e., the dynamic treatment effects) for $k,k'\in\mathcal{K}$, and
if we can, what the sign is.

\subsection{Data-Generating Framework\label{subsec:Data-Generating-Framework}}

We introduce a simple data-generating framework and formally define
the identified set. First, we introduce latent state variables that
generate $(\boldsymbol{Y},\boldsymbol{D})$. A latent state of the
world will determine specific maps $(\boldsymbol{y}^{t-1},\boldsymbol{d}^{t})\mapsto y_{t}$
and $(\boldsymbol{y}^{t-1},\boldsymbol{d}^{t-1},\boldsymbol{z}^{t})\mapsto d_{t}$
for $t=1,...,T$ under the exclusion restriction in Assumption SX.
We introduce the latent state variable $\tilde{S}_{t}$ whose realization
represents such a state. We define $\tilde{S}_{t}$ as follows. For
given $(\boldsymbol{y}^{t-1},\boldsymbol{d}^{t},\boldsymbol{z}^{t})$,
let $Y_{t}(\boldsymbol{y}^{t-1},\boldsymbol{d}^{t})$ and $D_{t}(\boldsymbol{y}^{t-1},\boldsymbol{d}^{t-1},\boldsymbol{z}^{t})$
denote the extended counterfactual outcomes and treatments, respectively,
and let $\{Y_{t}(\boldsymbol{y}^{t-1},\boldsymbol{d}^{t})\}$ and
$\{D_{t}(\boldsymbol{y}^{t-1},\boldsymbol{d}^{t-1},\boldsymbol{z}^{t})\}$
and their sequences w.r.t. $(\boldsymbol{y}^{t-1},\boldsymbol{d}^{t},\boldsymbol{z}^{t})$.
Then, by concatenating the two sequences, define $\tilde{S}_{t}\equiv(\{Y_{t}(\boldsymbol{y}^{t-1},\boldsymbol{d}^{t})\},\{D_{t}(\boldsymbol{y}^{t-1},\boldsymbol{d}^{t-1},\boldsymbol{z}^{t})\})\in\{0,1\}^{2^{2t-1}}\times\{0,1\}^{2^{3t-2}}$.
For example, $\tilde{S}_{1}=(Y_{1}(0),Y_{1}(1),D_{1}(0),D_{1}(1))\in\{0,1\}^{2}\times\{0,1\}^{2}$,
whose realization specifies particular maps $d_{1}\mapsto y_{1}$
and $z_{1}\mapsto d_{1}$. It is convenient to transform $\tilde{\boldsymbol{S}}\equiv(\tilde{S}_{1},...,\tilde{S}_{T})$
into a scalar (discrete) latent variable in $\mathbb{N}$ as $S\equiv\beta(\tilde{\boldsymbol{S}})\in\mathcal{S}\subset\mathbb{N}$,
where $\beta(\cdot)$ is a one-to-one map that transforms a binary
sequence into a decimal value. Define 
\begin{align*}
q_{s} & \equiv\Pr[S=s],
\end{align*}
and define the vector $q$ of $q_{s}$ which represents the distribution
of $S$, namely the true data-generating process. The vector $q$
resides in $\mathcal{Q}\equiv\{q:\sum_{s}q_{s}=1\text{ and }q_{s}\ge0\text{ }\forall s\}$
of dimension $d_{q}-1$ where $d_{q}\equiv\dim(q)$. A useful fact
is that the joint distribution of counterfactuals can be written as
a linear functional of $q$: 
\begin{align}
\Pr[\boldsymbol{Y}(\boldsymbol{d})=\boldsymbol{y},\boldsymbol{D}(\boldsymbol{z})=\boldsymbol{d}] & =\Pr[S\in\mathcal{S}:\boldsymbol{Y}(\boldsymbol{y}^{T-1},\boldsymbol{d})=\boldsymbol{y},\boldsymbol{D}(\boldsymbol{y}^{T-1},\boldsymbol{d}^{T-1},\boldsymbol{z})=\boldsymbol{d}]\nonumber \\
 & =\Pr[S\in\mathcal{S}:Y_{t}(\boldsymbol{y}^{t-1},\boldsymbol{d}^{t})=y_{t},D_{t}(\boldsymbol{y}^{t-1},\boldsymbol{d}^{t-1},\boldsymbol{z}^{t})=d_{t}\quad\forall t]\nonumber \\
 & =\sum_{s\in\mathcal{S}_{\boldsymbol{y},\boldsymbol{d}|\boldsymbol{z}}}q_{s},\label{eq:dist_counterfactuals}
\end{align}
where $\mathcal{S}_{\boldsymbol{y},\boldsymbol{d}|\boldsymbol{z}}$
is constructed by using the definition of $S$; its expression can
be found in Appendix \ref{sec:Matrices}.

Based on \eqref{eq:dist_counterfactuals}, the counterfactual welfare
can be written as a linear combination of $q_{s}$'s. That is, there
exists $1\times d_{q}$ vector $A_{k}$ of $1$'s and $0$'s such
that
\begin{align}
W_{k} & =A_{k}q.\label{eq:LP_welfare}
\end{align}
The formal derivation of $A_{k}$ can be found in Appendix \ref{sec:Matrices},
but the intuition is as follows. Recall $W_{k}\equiv f(q_{\boldsymbol{\delta}_{k}})$
where $q_{\boldsymbol{\delta}}(\boldsymbol{y})\equiv\Pr[\boldsymbol{Y}(\boldsymbol{\delta}(\cdot))=\boldsymbol{y}]$.
The key observation in deriving the result \eqref{eq:LP_welfare}
is that $\Pr[\boldsymbol{Y}(\boldsymbol{\delta}(\cdot))=\boldsymbol{y}]$
can be written as a linear functional of the joint distributions of
counterfactual outcomes with a \textit{static} regime, i.e., $\Pr[\boldsymbol{Y}(\boldsymbol{d})=\boldsymbol{y}]$'s,
which in turn is a linear functional of $q$. To illustrate with $T=2$
and welfare $W_{\boldsymbol{\delta}}=E[Y_{2}(\boldsymbol{\delta}(\cdot))]$,
we have 
\begin{align*}
\Pr[Y_{2}(\boldsymbol{\delta}(\cdot))=1] & =\sum_{y_{1}\in\{0,1\}}\Pr[Y_{2}(\delta_{1},\delta_{2}(Y_{1}(\delta_{1}),\delta_{1}))=1|Y_{1}(\delta_{1})=y_{1}]\Pr[Y_{1}(\delta_{1})=y_{1}]
\end{align*}
by the law of iterated expectation. Then, for instance, Regime 4 in
Table \ref{tab:regimes} yields 
\begin{align}
\Pr[Y_{2}(\boldsymbol{\delta}_{4}(\cdot))=1] & =P[\boldsymbol{Y}(1,1)=(1,1)]+P[\boldsymbol{Y}(1,0)=(0,1)],\label{eq:ex_regime8}
\end{align}
where each $\Pr[\boldsymbol{Y}(d_{1},d_{2})=(y_{1},y_{2})]$ is the
counterfactual distribution with a static regime, which in turn is
a linear functional of \eqref{eq:dist_counterfactuals}.

The data impose restrictions on $q\in\mathcal{Q}$. Define 
\begin{align*}
p_{\boldsymbol{y},\boldsymbol{d}|\boldsymbol{z}} & \equiv p(\boldsymbol{y},\boldsymbol{d}|\boldsymbol{z})\equiv\Pr[\boldsymbol{Y}=\boldsymbol{y},\boldsymbol{D}=\boldsymbol{d}|\boldsymbol{Z}=\boldsymbol{z}],
\end{align*}
and $p$ as the vector of $p_{\boldsymbol{y},\boldsymbol{d}|\boldsymbol{z}}$'s
except redundant elements. Let $d_{p}\equiv\dim(p)$. Since $\Pr[\boldsymbol{Y}=\boldsymbol{y},\boldsymbol{D}=\boldsymbol{d}|\boldsymbol{Z}=\boldsymbol{z}]=\Pr[\boldsymbol{Y}(\boldsymbol{d})=\boldsymbol{y},\boldsymbol{D}(\boldsymbol{z})=\boldsymbol{d}]$
by Assumption SX, we can readily show by \eqref{eq:dist_counterfactuals}
that there exists $d_{p}\times d_{q}$ matrix $B$ such that 
\begin{align}
Bq & =p,\label{eq:LP_constraint}
\end{align}
where each row of $B$ is a vector of $1$'s and $0$'s; the formal
derivation of $B$ can be found in Appendix \ref{sec:Matrices}. It
is worth noting that the linearity in \eqref{eq:LP_welfare} and \eqref{eq:LP_constraint}
is \textit{not} a restriction but given by the discrete nature of
the setting. We assume $rank(B)=d_{p}$ without loss of generality,
because redundant constraints do not play a role in restricting $\mathcal{Q}$.
We focus on the non-trivial case of $d_{p}<d_{q}$. If $d_{p}\ge d_{q}$,
which is rare, we can solve for $q=(B^{\top}B)^{-1}B^{\top}p$, and
can trivially point identify $W_{k}=A_{k}q$ and thus $\boldsymbol{\delta}^{*}(\cdot)$.
Otherwise, we have a set of observationally equivalent $q$'s, which
is the source of partial identification and motivates the following
definition of the identified set.\footnote{For simplicity, we use the same notation for the true $q$ and its
observational equivalence.}

For a given $q$, let $\boldsymbol{\delta}^{*}(\cdot;q)\equiv\arg\max_{\boldsymbol{\delta}_{k}(\cdot)\in\mathcal{D}}W_{k}=A_{k}q$
be the optimal regime, explicitly written as a function of the data-generating
process. \begin{definition}Under Assumption SX, the identified set
of $\boldsymbol{\delta}^{*}(\cdot)$ given the data distribution $p$
is 
\begin{align}
\mathcal{D}_{p}^{*} & \equiv\{\boldsymbol{\delta}^{*}(\cdot;q):Bq=p\text{ and }q\in\mathcal{Q}\}\subset\mathcal{D},\label{eq:ID_set}
\end{align}
which is assumed to be empty when $Bq\neq p$.\end{definition}

\subsection{Characterizing Partial Ordering and the Identified Set\label{subsec:Establishing-Parial-Ordering}}

Given $p$, we establish the partial ordering of $W_{k}$'s by determining
whether $W_{k}>W_{k'}$, $W_{k}<W_{k'}$, or $W_{k}$ and $W_{k'}$
are not comparable (including $W_{k}=W_{k'}$), denoted as $W_{k}\sim W_{k'}$,
for $k,k'\in\mathcal{K}$. As described in the next theorem, this
procedure can be accomplished by determining the signs of the bounds
on the welfare gap $W_{k}-W_{k'}$ for $k,k'\in\mathcal{K}$ and $k>k'$.\footnote{Note that directly comparing sharp bounds on welfares themselves will
\textit{not} deliver sharp partial ordering.} Then the identified set can be characterized based on the resulting
partial ordering.

The nature of the data generation induces the linear system \eqref{eq:LP_welfare}
and \eqref{eq:LP_constraint}. This enables us to characterize the
bounds on $W_{k}-W_{k'}=(A_{k}-A_{k'})q$ as the optima in linear
programming. Let $U_{k,k'}$ and $L_{k,k'}$ be the upper and lower
bounds. Also let $\Delta_{k,k'}\equiv A_{k}-A_{k'}$ for simplicity,
and thus the welfare gap is expressed as $W_{k}-W_{k'}=\Delta_{k,k'}q$.
Then, for $k,k'\in\mathcal{K}$, we have the main linear programs:
\begin{align}
\begin{array}{c}
U_{k,k'}=\max_{q\in\mathcal{Q}}\Delta_{k,k'}q,\\
L_{k,k'}=\min_{q\in\mathcal{Q}}\Delta_{k,k'}q,
\end{array} & \qquad s.t.\quad Bq=p.\label{eq:LP}
\end{align}
\begin{asB}$\{q:Bq=p\}\cap\mathcal{Q}\neq\emptyset$.\end{asB}

Assumption B imposes that the model, Assumption SX in this case, is
correctly specified. Under misspecification, the identified set is
empty by definition. The next theorem constructs the sharp partial
ordering and characterize the identified set using $U_{k,k'}$ and
$L_{k,k'}$ for $k,k'\in\mathcal{K}$ and $k>k'$, or equivalently,
$L_{k,k'}$ for $k,k'\in\mathcal{K}$ and $k\neq k'$.\footnote{Notice that $(L_{k,k'},U_{k,k'})$ for $k>k'$ contain the same information
as $L_{k,k'}$ for $k\neq k'$, since $U_{k,k'}=-L_{k',k}$.}

\begin{theorem}\label{thm:DAG}Suppose Assumptions SX and B hold.
Then, (i) $G(\mathcal{K},\mathcal{E}_{p})$ with $\mathcal{E}_{p}\equiv\{(k,k')\in\mathcal{K}:L_{k,k'}>0\text{ and }k\neq k'\}$
is sharp; (ii) $\mathcal{D}_{p}^{*}$ defined in \eqref{eq:ID_set}
satisfies 
\begin{align}
\mathcal{D}_{p}^{*} & =\{\boldsymbol{\delta}_{k'}(\cdot):\nexists k\in\mathcal{K}\text{ such that }L_{k,k'}>0\text{ and }k\neq k'\}\label{eq:char_ID_set}\\
 & =\{\boldsymbol{\delta}_{k'}(\cdot):L_{k,k'}\le0\text{ for all }k\in\mathcal{K}\text{ and }k\neq k'\},\label{eq:char_ID_set2}
\end{align}
and therefore the sets on the right-hand side are sharp.

\end{theorem}

The proof of Theorem \ref{thm:DAG} is shown in the Appendix. The
key insight of the proof is that even though the bounds on the welfare
gaps are calculated from separate optimizations, the partial ordering
is governed by \textit{common} $q$'s (each of which generates all
the welfares) that are observationally equivalent; see Section \ref{subsec:Topological-Sorting}
for related discussions.

Theorem \ref{thm:DAG}(i) prescribes how to calculate the sharp partial
ordering as a function of data.\footnote{The associated DAG can be conveniently represented in terms of a $\left|\mathcal{K}\right|\times\left|\mathcal{K}\right|$
adjacency matrix $\Omega$ such that its element $\Omega_{k,k'}=1$
if $W_{k}\ge W_{k'}$ and $\Omega_{k,k'}=0$ otherwise.} According to \eqref{eq:char_ID_set} in (ii), $\mathcal{D}_{p}^{*}$
is characterized as the collection of $\boldsymbol{\delta}_{k}(\cdot)$
where $k$ is in the set of \textit{maximal elements} of the partially
ordered set $G(\mathcal{K},\mathcal{E}_{p})$, i.e., the set of regimes
that are \textit{not inferior}. In Figure \ref{fig:partial_order},
it is easy to see that the set of maximals is $\mathcal{D}_{p}^{*}=\{\boldsymbol{\delta}_{1}(\cdot),\boldsymbol{\delta}_{4}(\cdot)\}$
in panel (a) and $\mathcal{D}_{p}^{*}=\{\boldsymbol{\delta}_{1}(\cdot)\}$
in panel (b). 

The identified set $\mathcal{D}_{p}^{*}$ characterizes the information
content of the model. Given the minimal structure we impose in the
model, $\mathcal{D}_{p}^{*}$ may be large in some cases. However,
we argue that an uninformative $\mathcal{D}_{p}^{*}$ still has implications
for policy: (i) such set may recommend the policymaker eliminate sub-optimal
regimes from her options;\footnote{Section \ref{sec:Discussions:-Estimation-and} discusses how to do
this systematically after embracing sampling uncertainty.} (ii) in turn, it warns the policymaker about her lack of information
(e.g., even if she has access to the experimental data); when $\mathcal{D}_{p}^{*}=\mathcal{D}$
as one extreme, ``no recommendation'' can be given as a non-trivial
policy suggestion of the need for better data. As shown in the numerical
exercise, the size of $\mathcal{D}_{p}^{*}$ is related to the strength
of $Z_{t}$ (i.e., the size of the complier group at $t$) and the
strength of the dynamic treatment effects. This is reminiscent of
the findings in \citet{machado2018instrumental} for the average treatment
effect in a static model. In Section \ref{sec:Additional-Assumptions},
we list further identifying assumptions that help shrink $\mathcal{D}_{p}^{*}$.

\section{Set of the $n$-th Best Regimes, Topological Sorts, and Bounds on
Sorted Welfare\label{sec:Topological-Sorting}}

In this section, we propose some ways to report results of this paper
including the partial ordering. These approaches can be useful especially
when the obtained partial ordering is complicated (e.g., with a longer
horizon).

\subsection{Set of the $n$-th Best Policies}

When the partial ordering of welfare is the parameter of interest,
the identified set of $\boldsymbol{\delta}^{*}(\cdot)$ can be viewed
as a summary of the partial ordering. This view can be extended to
introduce a set of the $n$-th best regimes, which further summarizes
the partial ordering. With slight abuse of notation, we can formalize
it as follows.

Recall $\mathcal{K}$ is the set of all regime indices. Motivated
from \eqref{eq:char_ID_set}, let $\mathcal{K}_{p}^{(1)}\equiv\{k':\nexists k\in\mathcal{K}\text{ such that }L_{k,k'}>0\text{ and }k\neq k'\in\mathcal{K}\}$
be the set of maximal elements of the partial ordering and let $\mathcal{D}_{p}^{(1)}\equiv\{\boldsymbol{\delta}_{k'}(\cdot):k'\in\mathcal{K}_{p}^{(1)}\}$.
Theorem \ref{thm:DAG}(ii) can be simply stated as $\mathcal{D}_{p}^{*}=\mathcal{D}_{p}^{(1)}$.
To define the set of second-best regimes, we first remove all the
elements in $\mathcal{K}_{p}^{(1)}$ from the set of candidate. Accordingly,
by defining
\begin{align*}
\mathcal{K}_{p}^{(2)} & \equiv\{k':\nexists k\in\mathcal{K}\backslash\mathcal{K}_{p}^{(1)}\text{ such that }L_{k,k'}>0\text{ and }k\neq k'\in\mathcal{K}\backslash\mathcal{K}_{p}^{(1)}\},
\end{align*}
we can introduce the set of second-best regimes: $\mathcal{D}_{p}^{(2)}\equiv\{\boldsymbol{\delta}_{k'}(\cdot):k'\in\mathcal{K}_{p}^{(2)}\}$.
Iteratively, we can define the set of $n$-th best regimes as $\mathcal{D}_{p}^{(n)}\equiv\{\boldsymbol{\delta}_{k'}(\cdot):k'\in\mathcal{K}_{p}^{(n)}\}$
where
\begin{align*}
\mathcal{K}_{p}^{(n)} & =\left\{ k':\nexists k\in\mathcal{K}\backslash\bigcup_{j=1}^{n-1}\mathcal{K}_{p}^{(j)}\text{ such that }L_{k,k'}>0\text{ and }k\neq k'\in\mathcal{K}\backslash\bigcup_{j=1}^{n-1}\mathcal{K}_{p}^{(j)}\right\} .
\end{align*}
The sets $\mathcal{D}_{p}^{(1)},...,\mathcal{D}_{p}^{(n)}$ are useful
policy benchmarks. For instance, the policy maker can conduct a sensitivity
analysis for her chosen regime (e.g., from a parametric model) by
inspecting in which set the regime is contained.

\subsection{Topological Sorts as Observational Equivalence\label{subsec:Topological-Sorting}}

Another way to summarize the partial ordering is to use topological
sorts. A \textit{topological sort} of a partial ordering is a linear
ordering of its vertices that does not violate the order in the partial
ordering. That is, for every directed edge $k\rightarrow k'$, $k$
comes before $k'$ in this linear ordering. Apparently, there can
be multiple topological sorts for a partial ordering. Let $L_{G}$
be the number of topological sorts of partial ordering $G(\mathcal{K},\mathcal{E}_{p})$,
and let $k_{l,1}\in\mathcal{K}$ be the initial vertex of the $l$-th
topological sort for $1\le l\le L_{G}$. For example, given the partial
ordering in Figure \ref{fig:partial_order}(a), $(\boldsymbol{\delta}_{1},\boldsymbol{\delta}_{4},\boldsymbol{\delta}_{2},\boldsymbol{\delta}_{3})$
is an example of a topological sort (with $k_{l,1}=1$), but $(\boldsymbol{\delta}_{1},\boldsymbol{\delta}_{2},\boldsymbol{\delta}_{4},\boldsymbol{\delta}_{3})$
is not. Topological sorts are routinely reported for a given partial
ordering, and there are well-known algorithms that efficiently find
topological sorts, such as \citet{kahn1962topological}'s algorithm.

In fact, topological sorts can be viewed as total orderings that are
\textit{observationally equivalent} to the true \textit{total} ordering
of welfares. That is, each $q$ generates the total ordering of welfares
via $W_{k}=A_{k}q$, and $q$'s in $\{q:Bq=p\}\cap\mathcal{Q}$ generates
observationally equivalent total orderings. This insight enables us
to interpret the partial ordering we establish using the more conventional
notion of partial identification: the ordering is partially identified
in the sense that the set of all topological sorts is not a singleton.
This insight yields an alternative way of characterizing the identified
set $\mathcal{D}_{p}^{*}$ of the optimal regime.

\begin{theorem}\label{thm:topo_sort_ID_set}Suppose Assumptions SX
and B hold. The identified set $\mathcal{D}_{p}^{*}$ defined in \eqref{eq:ID_set}
satisfies 
\begin{align*}
\mathcal{D}_{p}^{*}=\{\boldsymbol{\delta}_{k_{l,1}}(\cdot):1\le l\le L_{G}\},
\end{align*}
where $k_{l,1}$ is the initial vertex of the $l$-th topological
sort of $G(\mathcal{K},\mathcal{E}_{p})$.

\end{theorem}

Suppose the partial ordering we recover from the data is not too
sparse. By definition, a topological sort provides a ranking of regimes
that is \textit{not inconsistent} with the partial welfare ordering.
Therefore, not only $\boldsymbol{\delta}_{k_{l,1}}(\cdot)\in\mathcal{D}_{p}^{*}$
but also the full sequence of a topological sort
\begin{align}
\left(\boldsymbol{\delta}_{k_{l,1}}(\cdot),\boldsymbol{\delta}_{k_{l,2}}(\cdot),...,\boldsymbol{d}_{k_{l,\left|\mathcal{D}\right|}}(\cdot)\right)\label{eq:top_sort}
\end{align}
can be useful. A policymaker can be equipped with any of such sequences
as a policy benchmark.

\subsection{Bounds on Sorted Welfares\label{subsec:Bounds-on-Sorted}}

The set of $n$-th best regimes and topological sorts provide ordinal
information about counterfactual welfares. To gain more comprehensive
knowledge about the welfares, they can be accompanied by cardinal
information: bounds on the sorted welfares. One might especially be
interested in the bounds on ``top-tier'' welfares that are associated
with the identified set or the first few elements in the topological
sort. Bounds on gains from adaptivity and regrets can also be computed.
These bounds can be calculated by solving linear programs. For instance,
the sharp lower and upper bounds on welfare $W_{k}$ can be calculated
via 
\begin{align}
\begin{array}{c}
U_{k}=\max_{q\in\mathcal{Q}}A_{k}q,\\
L_{k}=\min_{q\in\mathcal{Q}}A_{k}q,
\end{array} & \qquad s.t.\quad Bq=p.\label{eq:LP-1}
\end{align}

\section{Additional Assumptions\label{sec:Additional-Assumptions}}

Often, researchers are willing to impose more assumptions based on
priors about the data-generating process, e.g., agent's behaviors.
Examples are uniformity, agent's learning, Markovian structure, and
stationarity. These assumptions are easy to incorporate within the
linear programming \eqref{eq:LP}. These assumptions tighten the
identified set $\mathcal{D}_{p}^{*}$ by reducing the dimension of
simplex $\mathcal{Q}$, and thus producing a denser partial ordering.\footnote{Similarly, when these assumptions are incorporated in \eqref{eq:LP-1},
we obtain tighter bounds on welfares.}

To incorporate these assumptions, we extend the framework introduced
in Sections \ref{sec:Partial-Ordering-and}--\ref{sec:Topological-Sorting}.
Suppose $h$ is a $d_{q}\times1$ vector of ones and zeros, where
zeros are imposed by given identifying assumptions. Introduce $d_{q}\times d_{q}$
diagonal matrix $H=diag(h)$. Then, we can define a standard simplex
for $\bar{q}\equiv Hq$ as 
\begin{align}
\bar{\mathcal{Q}} & \equiv\{\bar{q}:\sum_{s}\bar{q}_{s}=1\text{ and }\bar{q}_{s}\ge0\text{ }\forall s\}.\label{eq:simplex2}
\end{align}
Note that the dimension of this simplex is smaller than the dimension
$d_{q}$ of $\mathcal{Q}$ if $h$ contains zeros. Then we can modify
\eqref{eq:LP_welfare} and \eqref{eq:LP_constraint} as 
\begin{align*}
B\bar{q} & =p,\\
W_{k} & =A_{k}\bar{q},
\end{align*}
respectively. Let $\boldsymbol{\delta}^{*}(\cdot;\bar{q})\equiv\arg\max_{\boldsymbol{\delta}_{k}(\cdot)\in\mathcal{D}}W_{k}=A_{k}\bar{q}$.
Then, the identified set with the identifying assumptions coded in
$h$ is defined as 
\begin{align}
\bar{\mathcal{D}}_{p}^{*} & \equiv\{\boldsymbol{\delta}^{*}(\cdot;\bar{q}):B\bar{q}=p\text{ and }\bar{q}\in\mathcal{Q}\}\subset\mathcal{D},\label{eq:ID_set-1}
\end{align}
which is assumed to be empty when $B\bar{q}\neq p$. Importantly,
the latter occurs when any of the identifying assumptions are misspecified.
Note that $H$ is idempotent. Define $\bar{\Delta}\equiv\Delta H$
and $\bar{B}\equiv BH$. Then $\Delta\bar{q}=\bar{\Delta}\bar{q}$
and $B\bar{q}=\bar{B}\bar{q}$. Therefore, to generate the partial
ordering and characterize the identified set, Theorem \ref{thm:DAG}
can be modified by replacing $q$, $B$ and $\Delta$ with $\bar{q}$,
$\bar{B}$ and $\bar{\Delta}$, respectively.

We now list examples of identifying assumptions. This list is far
from complete, and there may be other assumptions on how $(\boldsymbol{Y},\boldsymbol{D},\boldsymbol{Z})$
are generated. The first assumption is a sequential version of the
uniformity assumption (i.e., the monotonicity assumption) in \citet{imbens1994identification}.\begin{asM1}For
each $t$, either $D_{t}(\boldsymbol{Z}^{t-1},1)\ge D_{t}(\boldsymbol{Z}^{t-1},0)$
w.p.1 or $D_{t}(\boldsymbol{Z}^{t-1},1)\le D_{t}(\boldsymbol{Z}^{t-1},0)$
w.p.1. conditional on $(\boldsymbol{Y}^{t-1},\boldsymbol{D}^{t-1},\boldsymbol{Z}^{t-1})$.\end{asM1}

Assumption M1 postulates that there is no defying (or complying) behavior
in decision $D_{t}$ conditional on $(\boldsymbol{Y}^{t-1},\boldsymbol{D}^{t-1},\boldsymbol{Z}^{t-1})$.
Without being conditional on $(\boldsymbol{Y}^{t-1},\boldsymbol{D}^{t-1},\boldsymbol{Z}^{t-1})$,
however, there can be a general non-monotonic pattern in the way that
$\boldsymbol{Z}^{t}$ influences $\boldsymbol{D}^{t}$. For example,
we can have $D_{t}(\boldsymbol{Z}^{t-1},1)\ge D_{t}(\boldsymbol{Z}^{t-1},0)$
for $D_{t-1}=1$ while $D_{t}(\boldsymbol{Z}^{t-1},1)\le D_{t}(\boldsymbol{Z}^{t-1},0)$
for $D_{t-1}=0$. Recall $\tilde{S}_{t}\equiv(\{Y_{t}(\boldsymbol{y}^{t-1},\boldsymbol{d}^{t})\},\{D_{t}(\boldsymbol{y}^{t-1},\boldsymbol{d}^{t-1},\boldsymbol{z}^{t})\})\in\{0,1\}^{2^{2t-1}}\times\{0,1\}^{2^{3t-2}}$.
For example, the no-defier assumption can be incorporated in $h$
by having $h_{s}=0$ for $s\in\{S=\beta(\tilde{\boldsymbol{S}}):D_{t}(\boldsymbol{y}^{t-1},\boldsymbol{d}^{t-1},\boldsymbol{z}^{t-1},1)=0\text{ and }D_{t}(\boldsymbol{y}^{t-1},\boldsymbol{d}^{t-1},\boldsymbol{z}^{t-1},0)=1\text{ }\forall t\}$
and $h_{s}=1$ otherwise. By extending the idea of \citet{vytlacil2002independence},
we can show that M1 is the equivalent of imposing a threshold-crossing
model for $D_{t}$: 
\begin{align}
D_{t} & =1\{\pi_{t}(\boldsymbol{Y}^{t-1},\boldsymbol{D}^{t-1},\boldsymbol{Z}^{t})\ge\nu_{t}\},\label{eq:model2_only}
\end{align}
where $\pi_{t}(\cdot)$ is an unknown, measurable, and non-trivial
function of $Z_{t}$.

\begin{lemma}\label{lem:vyt_1}Suppose Assumption SX holds and $\Pr[D_{t}=1|\boldsymbol{Y}^{t-1},\boldsymbol{D}^{t-1},\boldsymbol{Z}^{t}]$
is a nontrivial function of $Z_{t}$. Assumption M1 is equivalent
to \eqref{eq:model2_only} being satisfied conditional on $(\boldsymbol{Y}^{t-1},\boldsymbol{D}^{t-1},\boldsymbol{Z}^{t-1})$
for each $t$.

\end{lemma}The dynamic selection model \eqref{eq:model2_only} should
not be confused with the dynamic regime \eqref{eq:trt_rule}. Compared
to the dynamic regime $d_{t}=\delta_{t}(\boldsymbol{y}^{t-1},\boldsymbol{d}^{t-1})$,
which is a hypothetical quantity, equation \eqref{eq:model2_only}
models each individual's observed treatment decision, in that it is
not only a function of $(\boldsymbol{Y}^{t-1},\boldsymbol{D}^{t-1})$
but also $\nu_{t}$, the individual's unobserved characteristics.
We assume that the policymaker has no access to $\boldsymbol{\nu}\equiv(\nu_{1},...,\nu_{T})$.
The functional dependence of $D_{t}$ on  ($\boldsymbol{Y}^{t-1},\boldsymbol{D}^{t-1}$)
and $\boldsymbol{Z}^{t-1}$ reflects the agent's learning. Indeed,
a specific version of such learning can be imposed as an additional
identifying assumption:

\begin{asL}For each $t$ and given $\boldsymbol{z}^{t}$, $D_{t}(\boldsymbol{y}^{t-1},\boldsymbol{d}^{t-1},\boldsymbol{z}^{t})\ge D_{t}(\tilde{\boldsymbol{y}}^{t-1},\tilde{\boldsymbol{d}}^{t-1},\boldsymbol{z}^{t})$
$\text{w.p.1}$ for $(\boldsymbol{y}^{t-1},\boldsymbol{d}^{t-1})$
and $(\tilde{\boldsymbol{y}}^{t-1},\tilde{\boldsymbol{d}}^{t-1})$
such that $\left\Vert \boldsymbol{y}^{t-1}-\boldsymbol{d}^{t-1}\right\Vert <\left\Vert \tilde{\boldsymbol{y}}^{t-1}-\tilde{\boldsymbol{d}}^{t-1}\right\Vert $
(long memory) or $y_{t-1}-d_{t-1}<\tilde{y}_{t-1}-\tilde{d}_{t-1}$
(short memory).\end{asL}

According to Assumption L, an agent has the ability to revise her
next period's decision based on her memory. To illustrate, consider
the second period decision, $D_{2}(y_{1},d_{1})$. Under Assumption
L, an agent who would switch her treatment decision at $t=2$ had
she experienced bad health ($y_{1}=0$) after receiving the treatment
($d_{1}=1$), i.e., $D_{2}(0,1)=0$, would remain to take the treatment
had she experienced good health, i.e., $D_{2}(1,1)=1$. Moreover,
if an agent has not switched even after bad health, i.e., $D_{2}(0,1)=1$,
it should be because of her unobserved preference, and thus $D_{2}(1,1)=1$,
not because she cannot learn from the past, i.e., $D_{2}(1,1)=0$
cannot happen.\footnote{As suggested in this example, Assumption L makes the most sense when
$Y_{t}$ and $D_{t}$ are the same (or at least similar) types over
time, which is not generally required for the analysis of this paper.}

Sometimes, we want to further impose uniformity in the formation of
$Y_{t}$ on top of Assumption M1:

\begin{asM2}Assumption M1 holds, and for each $t$, either $Y_{t}(\boldsymbol{D}^{t-1},1)\ge Y_{t}(\boldsymbol{D}^{t-1},0)$
w.p.1 or $Y_{t}(\boldsymbol{D}^{t-1},1)\le Y_{t}(\boldsymbol{D}^{t-1},0)$
w.p.1 conditional on $(\boldsymbol{Y}^{t-1},\boldsymbol{D}^{t-1})$.\end{asM2}

This assumption postulates uniformity in a way that restricts heterogeneity
of the contemporaneous treatment effect. However, similarly as before,
without being conditional on $(\boldsymbol{Y}^{t-1},\boldsymbol{D}^{t-1})$,
there can be a general non-monotonic pattern in the way that $\boldsymbol{D}^{t}$
influences $\boldsymbol{Y}^{t}$. For example, we can have $Y_{t}(\boldsymbol{D}^{t-1},1)\ge Y_{t}(\boldsymbol{D}^{t-1},0)$
for $Y_{t-1}=1$ while $Y_{t}(\boldsymbol{D}^{t-1},1)\le Y_{t}(\boldsymbol{D}^{t-1},0)$
for $Y_{t-1}=0$. It is also worth noting that Assumption M2 (and
M1) does not assume the direction of monotonicity, but the direction
is recovered from the data. This is in contrast to the monotone treatment
response assumption in, e.g., \citet{manski1997monotone} and \citet{manski2000monotone},
which assume the direction. Using a similar argument as before, Assumption
M2 is the equivalent of a dynamic version of a nonparametric triangular
model: 
\begin{align}
Y_{t} & =1\{\mu_{t}(\boldsymbol{Y}^{t-1},\boldsymbol{D}^{t})\ge\varepsilon_{t}\},\label{eq:model1}\\
D_{t} & =1\{\pi_{t}(\boldsymbol{Y}^{t-1},\boldsymbol{D}^{t-1},\boldsymbol{Z}^{t})\ge\nu_{t}\},\label{eq:model2}
\end{align}
where $\mu_{t}(\cdot)$ and $\pi_{t}(\cdot)$ are unknown, measurable,
and non-trivial functions of $D_{t}$ and $Z_{t}$, respectively.

\begin{lemma}\label{lem:vyt_2}Suppose Assumption SX holds, $\Pr[D_{t}=1|\boldsymbol{Y}^{t-1},\boldsymbol{D}^{t-1},\boldsymbol{Z}^{t}]$
is a non-trivial function of $Z_{t}$, and $\Pr[Y_{t}=1|\boldsymbol{Y}^{t-1},\boldsymbol{D}^{t}]$
is a non-trivial function of $D_{t}$. Assumption M2 is equivalent
to \eqref{eq:model1}--\eqref{eq:model2} being satisfied conditional
on $(\boldsymbol{Y}^{t-1},\boldsymbol{D}^{t-1},\boldsymbol{Z}^{t-1})$
for each $t$.

\end{lemma}

The next assumption imposes a Markov-type structure in the $Y_{t}$
and $D_{t}$ processes.

\begin{asK}$Y_{t}|(\boldsymbol{Y}^{t-1},\boldsymbol{D}^{t})\stackrel{d}{=}Y_{t}|(Y_{t-1},D_{t})$
and $D_{t}|(\boldsymbol{Y}^{t-1},\boldsymbol{D}^{t-1},\boldsymbol{Z}^{t})\stackrel{d}{=}D_{t}|(Y_{t-1},D_{t-1},Z_{t})$
for each $t$.\end{asK}

In terms of the triangular model \eqref{eq:model1}--\eqref{eq:model2},
Assumption K implies 
\begin{align*}
Y_{t} & =1\{\mu_{t}(Y_{t-1},D_{t})\ge\varepsilon_{t}\},\\
D_{t} & =1\{\pi_{t}(Y_{t-1},D_{t-1},Z_{t})\ge\nu_{t}\},
\end{align*}
which yields the familiar structure of dynamic discrete choice models
found in the literature. Lastly, when there are more than two periods,
an assumption that imposes stationarity can be helpful for identification.
Such an assumption can be found in \citet{torgovitsky2016partial}.

\section{Cardinality Reduction\label{sec:Cardinality-Reduction}}

The typical time horizons we consider in this paper are short. For
example, a multi-stage experiment called the Fast Track Prevention
Program (\citet*{conduct1992developmental}) considers $T=4$. When
$T$ is not small, the cardinality of $\mathcal{D}$ may be too large,
and we may want to reduce it for computational, institutional, and
practical purposes.

One way to reduce the cardinality is to reduce the dimension of the
adaptivity. Define a simpler adaptive treatment rule $\delta_{t}:\{0,1\}\times\{0,1\}\rightarrow\{0,1\}$
that maps only the lagged outcome and treatment onto a treatment allocation
$d_{t}\in\{0,1\}$: 
\begin{align*}
\delta_{t}(y_{t-1},d_{t-1}) & =d_{t}.
\end{align*}
In this case, we have $\left|\mathcal{D}\right|=2^{2T-1}$ instead
of $2^{2^{T}-1}$. An even simpler rule, $\delta_{t}(y_{t-1})$, appears
in \citet{murphy2001marginal}.

Another possibility is to be motivated by institutional or budget
constraints. For example, it may be the case that adaptive allocation
is available every second period or only later in the horizon due
to cost considerations. For example, suppose that the policymaker
decides to introduce the adaptive rule at $t=T$ while maintaining
static rules for $t\le T-1$. Finally, $\mathcal{D}$ can be restricted
by budget or policy constraints that, e.g., the treatment is allocated
to each individual at most once.

\section{Numerical Studies\label{sec:Numerical-Studies}}

We conduct numerical exercises to illustrate (i) the theoretical results
developed in Sections \ref{sec:Partial-Ordering-and}--\ref{sec:Topological-Sorting},
(ii) the role of the assumptions introduced in Section \ref{sec:Additional-Assumptions},
and (iii) the overall computational scale of the problem. For $T=2$,
we consider the following data-generating process: 
\begin{align}
D_{i1} & =1\{\pi_{1}Z_{i1}+\alpha_{i}+v_{i1}\ge0\},\label{eq:dgp1}\\
Y_{i1} & =1\{\mu_{1}D_{i1}+\alpha_{i}+e_{i1}\ge0\},\label{eq:dgp2}\\
D_{i2} & =1\{\pi_{21}Y_{i1}+\pi_{22}D_{i1}+\pi_{23}Z_{i2}+\alpha_{i}+v_{i2}\ge0\},\label{eq:dgp3}\\
Y_{i2} & =1\{\mu_{21}Y_{i1}+\mu_{22}D_{i2}+\alpha_{i}+e_{i2}\ge0\},\label{eq:dgp4}
\end{align}
where $(v_{1},e_{1},v_{2},e_{2},\alpha)$ are mutually independent
and jointly normally distributed, the endogeneity of $D_{i1}$ and
$D_{i2}$ as well as the serial correlation of the unobservables are
captured by the individual effect $\alpha_{i}$, and $(Z_{1},Z_{2})$
are Bernoulli, independent of $(v_{1},e_{1},v_{2},e_{2},\alpha)$.
Notice that the process is intended to satisfy Assumptions SX, K,
M1, and M2. We consider a data-generating process where all the coefficients
in \eqref{eq:dgp1}--\eqref{eq:dgp4} take positive values. In this
exercise, we consider the welfare $W_{k}=E[Y_{2}(\boldsymbol{\delta}_{k}(\cdot))]$.
\begin{figure}
\begin{centering}
\includegraphics[scale=0.3]{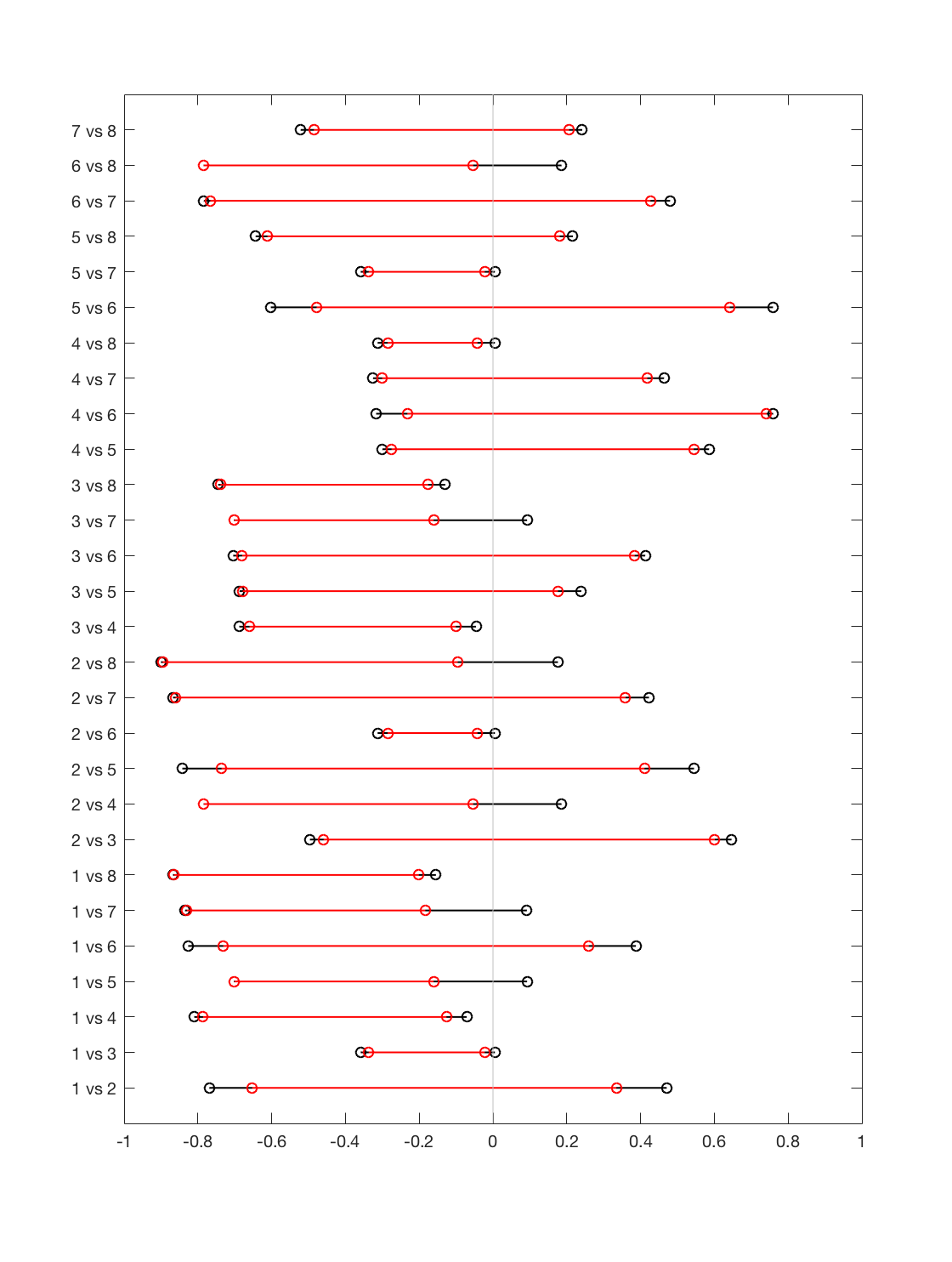}
\par\end{centering}
\caption{Sharp Bounds on Welfare Gaps under M1 (black) and M2 (red)}
\label{fig:bd_gap} 
\end{figure}
\begin{figure}
\begin{centering}
\includegraphics[scale=0.28]{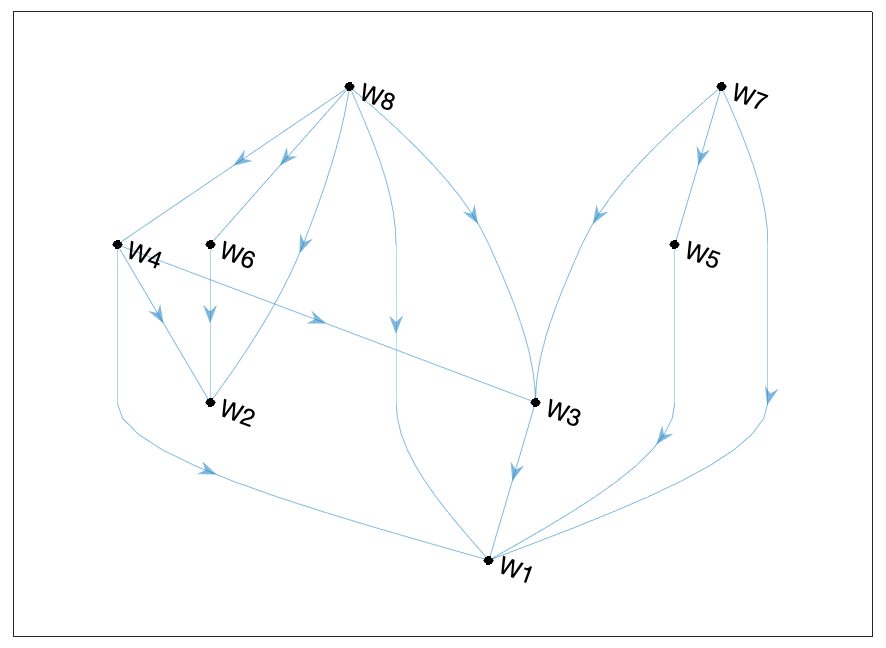}
\par\end{centering}
\caption{Sharp Partial Ordering under M2}
\label{fig:dag}
\end{figure}

As shown in Table \ref{tab:regimes}, there are eight possible regimes,
i.e., $\left|\mathcal{\mathcal{D}}\right|=\left|\mathcal{K}\right|=8$.
We calculate the lower and upper bounds $(L_{k,k'},U_{k,k'})$ on
the welfare gap $W_{k}-W_{k'}$ for all pairs $k,k'\in\{1,...,8\}$
($k<k'$). This is to illustrate the role of assumptions in improving
the bounds. We conduct the bubble sort, which makes {\scriptsize{}$\left(\begin{array}{c}
8\\
2
\end{array}\right)$}$=28$ pair-wise comparisons, resulting in $28\times2$ linear programs
to run.\footnote{There are more efficient algorithms than the bubble sort, such as
the \textit{quick sort}, although they must be modified to incorporate
the distinct feature of our problem: the possible incomparability
that stems from partial identification. Note that for comparable pairs,
transitivity can be applied and thus the total number of comparisons
can be smaller.} As the researcher, we maintain Assumption K. Then, for each linear
program, the dimension of $q$ is $\left|\mathcal{Q}\right|+1=\left|\mathcal{S}\right|=\left|\mathcal{S}_{1}\right|\times\left|\mathcal{S}_{2}\right|=2^{2}\times2^{2}\times2^{8}\times2^{4}=65,536$.\footnote{The dimension is reduced with additional identifying assumptions.}
The number of main constraints is $\dim(p)=2^{3\times2}-2^{2}=60$.
There are $1+65,536$ additional constraints that define the simplex,
i.e., $\sum_{s}q_{s}=1$ and $q_{s}\ge0$ for all $s\in\mathcal{S}$.
Each linear program takes less than a second to calculate $L_{k,k'}$
or $U_{k,k'}$ with a computer with a 2.2 GHz single-core processor
and 16 GB memory and with a modern solver such as CPLEX, MOSEK, and
GUROBI.

Figure \ref{fig:bd_gap} reports the bounds $(L_{k,k'},U_{k,k'})$
on $W_{k}-W_{k'}$ for all $(k,k')\in\{1,...,8\}$ under Assumption
M1 (in black) and Assumption M2 (in red). In the figure, we can determine
the sign of the welfare gap for those bounds that exclude zero. The
difference between the black and red bounds illustrates the role of
Assumption M2 relative to M1. That is, there are more bounds that
avoid the zero vertical line with M2, which is consistent with the
theory. It is important to note that, because M2 does not assume
the direction of monotonicity, the sign of the welfare gap is not
imposed by the assumption but recovered from the data.\footnote{The direction of the monotonicity in M2 can be estimated directly
from the data by using the fact that $\text{sign}(E[Y_{t}|Z_{t}=1,\boldsymbol{Y}^{t-1},\boldsymbol{D}^{t-1}]-E[Y_{t}|Z_{t}=1,\boldsymbol{Y}^{t-1},\boldsymbol{D}^{t-1}])=\text{sign}(E[Y_{t}(D^{t-1},1)|\boldsymbol{Y}^{t-1},\boldsymbol{D}^{t-1}]-E[Y_{t}(D^{t-1},0)|,\boldsymbol{Y}^{t-1},\boldsymbol{D}^{t-1}])$
almost surely. This result is an extension of \citet{SV11} to our
multi-period setting.} Each set of bounds generates an associated partial ordering as a
DAG (produced as an $8\times8$ adjacency matrix).\footnote{Given the solutions of the linear programs, the adjacency matrix and
thus the graph is simple to produce automatically using a standard
software such as MATLAB.} We proceed with Assumption M2 for brevity.

Figure \ref{fig:dag} (identical to Figure \ref{fig:dag-1} in the
Introduction) depicts the sharp partial ordering generated from $(L_{k,k'},U_{k,k'})$'s
under Assumption M2, based on Theorem \ref{thm:DAG}(i). Then, by
Theorem \ref{thm:DAG}(ii), the identified set of $\boldsymbol{\delta}^{*}(\cdot)$
is 
\begin{align*}
\mathcal{D}_{p}^{*} & =\{\boldsymbol{\delta}_{7}(\cdot),\boldsymbol{\delta}_{8}(\cdot)\}.
\end{align*}
The common feature of the elements in $\mathcal{D}_{p}^{*}$ is that
it is optimal to allocate $\delta_{2}=1$ for all $y_{1}\in\{0,1\}$.
Finally, the following is one of the topological sorts produced from
the partial ordering: 
\begin{align*}
(\boldsymbol{\delta}_{8}(\cdot),\boldsymbol{\delta}_{4}(\cdot),\boldsymbol{\delta}_{7}(\cdot),\boldsymbol{\delta}_{3}(\cdot),\boldsymbol{\delta}_{5}(\cdot),\boldsymbol{\delta}_{1}(\cdot),\boldsymbol{\delta}_{6}(\cdot),\boldsymbol{\delta}_{2}(\cdot)).
\end{align*}

We also conducted a parallel analysis but with a slightly different
data-generating process, where (a) all the coefficients in \eqref{eq:dgp1}--\eqref{eq:dgp4}
are positive except $\mu_{22}<0$ and (b) $Z_{2}$ does not exist.
In Case (a), we obtain $\mathcal{D}_{p}^{*}=\{\boldsymbol{\delta}_{2}(\cdot)\}$
as a singleton, i.e., we point identify $\boldsymbol{\delta}^{*}(\cdot)=\boldsymbol{\delta}_{2}(\cdot)$.
The partial ordering for Case (b) is shown in Figure \ref{fig:dag-2}.
In this case, we obtain $\mathcal{D}_{p}^{*}=\{\boldsymbol{\delta}_{6}(\cdot),\boldsymbol{\delta}_{7}(\cdot),\boldsymbol{\delta}_{8}(\cdot)\}$.
\begin{figure}
\begin{centering}
\includegraphics[scale=0.28]{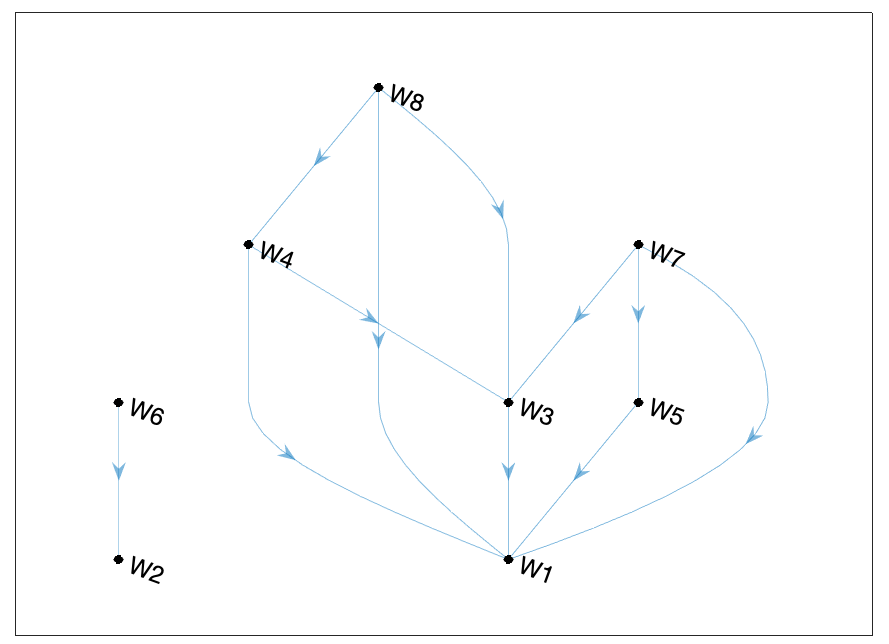}
\par\end{centering}
\caption{Sharp Partial Ordering under M2 (with only $Z_{1}$)}
\label{fig:dag-2}
\end{figure}

\section{Application\label{sec:Application}}

We apply the framework of this paper to understand returns to schooling
and post-school training as a sequence of treatments and to conduct
a policy analysis. Schooling and post-school training are two major
interventions that affect various labor market outcomes, such as earnings
and employment status (\citet{ashenfelter2010handbook}). These treatments
also have influences on health outcomes, either directly or through
the labor market outcomes, and thus of interest for public health
policies (\citet{backlund1996shape}, \citet{mcdonough1997income},
\citet{case2002economic}). We find that the Job Training Partnership
Act (JTPA) is an appropriate setting for our analysis. The JTPA program
is one of the largest publicly-funded training programs in the United
States for economically disadvantaged individuals. Unfortunately,
the JTPA only concerns post-school trainings, which have been the
main focus in the literature (\citet{bloom1997benefits}, \citet{abadie2002instrumental},
\citet{kitagawa2018should}). In this paper, we combine the JTPA Title
II data with those from other sources regarding high school education
to create a data set that allows us to study the effects of a high
school (HS) diploma (or its equivalents) and the subsidized job trainings
as a sequence of treatments. We consider high school diplomas rather
than college degrees because the former is more relevant for the disadvantaged
population of Title II of the JTPA program.

We are interested in the dynamic treatment regime $\boldsymbol{\delta}(\cdot)=(\delta_{1},\delta_{2}(\cdot))$,
where $\delta_{1}$ is a HS diploma and $\delta_{2}(y_{1})$ is the
job training program given pre-program earning type $y_{1}$. The
motivation of having $\delta_{2}$ as a function of $y_{1}$ comes
from acknowledging  the dynamic nature of how earnings are formed
under education and training. The first-stage allocation $\delta_{1}$
will affect the pre-program earning. This response may contain information
about unobserved characteristics of the individuals. Therefore, the
allocation of $\delta_{2}$ can be informed by being adaptive to $y_{1}$.
Then, the counterfactual earning type in the terminal stage given
$\boldsymbol{\delta}(\cdot)$ can be expressed as $Y_{2}(\boldsymbol{\delta}(\cdot))=Y_{2}(\delta_{1},\delta_{2}(Y_{1}(\delta_{1})))$
where $Y_{1}(\delta_{1})$ is the counterfactual earning type in the
first stage given $\delta_{1}$. We are interested in the optimal
regime $\boldsymbol{\delta}^{*}$ that maximizes each of the following
welfares: the average terminal earning $E[Y_{2}(\boldsymbol{\delta}(\cdot))]$
and the average lifetime earning $E[Y_{1}(\delta_{1})]+E[Y_{2}(\boldsymbol{\delta}(\cdot))]$.

For the purpose of our analysis, we combine the JTPA data with data
from the US Census and the National Center for Education Statistics
(NCES), from which we construct the following set of variables: $Y_{2}$
above or below median of 30-month earnings, $D_{2}$ the job training
program, $Z_{2}$ a random assignment of the program, $Y_{1}$ above
or below 80th percentile of pre-program earnings, $D_{1}$ the HS
diploma or GED, and $Z_{1}$ the number of high schools per square
mile.\footnote{For $Y_{1}$, the 80th percentile cutoff is chosen as it is found
to be relevant in defining subpopulations that have contrasting effects
of the program. There are other covariates in the constructed dataset,
but we omit them for the simplicity of our analysis. These variables
can be incorporated as pre-treatment covariates so that the first-stage
treatment is adaptive to them.} The instrument $Z_{1}$ for the HS treatment appears in the literature
(e.g., \citet{neal1997effects}). The number of individuals in the
sample is 9,223. We impose Assumptions SX and M2 throughout the analysis.

\begin{figure}
\begin{centering}
\includegraphics[scale=0.5]{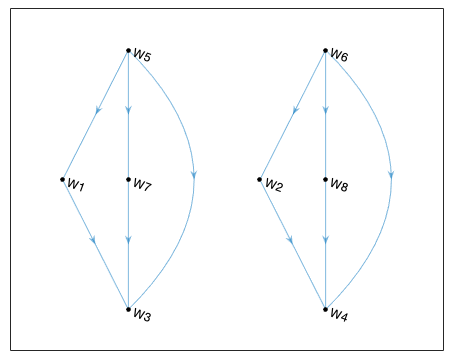}~{\footnotesize{}}%
\begin{tabular}[b]{|c|c|c|c|}
\hline 
{\footnotesize{}Regime \#} & {\footnotesize{}$\delta_{1}$} & {\footnotesize{}$\delta_{2}(1,\delta_{1})$} & {\footnotesize{}$\delta_{2}(0,\delta_{1})$}\tabularnewline
\hline 
\hline 
{\footnotesize{}1} & {\footnotesize{}0} & {\footnotesize{}0} & {\footnotesize{}0}\tabularnewline
\hline 
{\footnotesize{}2} & {\footnotesize{}1} & {\footnotesize{}0} & {\footnotesize{}0}\tabularnewline
\hline 
{\footnotesize{}3} & {\footnotesize{}0} & {\footnotesize{}1} & {\footnotesize{}0}\tabularnewline
\hline 
{\footnotesize{}4} & {\footnotesize{}1} & {\footnotesize{}1} & {\footnotesize{}0}\tabularnewline
\hline 
\textcolor{red}{\footnotesize{}5} & \textcolor{red}{\footnotesize{}0} & \textcolor{red}{\footnotesize{}0} & \textcolor{red}{\footnotesize{}1}\tabularnewline
\hline 
\textcolor{red}{\footnotesize{}6} & \textcolor{red}{\footnotesize{}1} & \textcolor{red}{\footnotesize{}0} & \textcolor{red}{\footnotesize{}1}\tabularnewline
\hline 
{\footnotesize{}7} & {\footnotesize{}0} & {\footnotesize{}1} & {\footnotesize{}1}\tabularnewline
\hline 
{\footnotesize{}8} & {\footnotesize{}1} & {\footnotesize{}1} & {\footnotesize{}1}\tabularnewline
\hline 
\multicolumn{1}{c}{} & \multicolumn{1}{c}{} & \multicolumn{1}{c}{} & \multicolumn{1}{c}{}\tabularnewline
\multicolumn{1}{c}{} & \multicolumn{1}{c}{} & \multicolumn{1}{c}{} & \multicolumn{1}{c}{}\tabularnewline
\end{tabular}{\footnotesize\par}
\par\end{centering}
\caption{Estimated Partial Ordering of $W_{\boldsymbol{\delta}}=E[Y_{2}(\boldsymbol{\delta}(\cdot))]$
and Estimated Set for $\boldsymbol{\delta}^{*}$ (red)}
\label{fig:jtpa}
\end{figure}

The estimation of the partial ordering (i.e., the DAG) and the identified
set $\mathcal{D}_{p}^{*}$ is straightforward given the conditions
in Theorem \ref{thm:DAG} and the linear programs \eqref{eq:LP}.
The only unknown object is $p$, the joint distribution of $(\boldsymbol{Y},\boldsymbol{D},\boldsymbol{Z})$,
which can be estimated as $\hat{p}$, a vector of $\hat{p}_{\boldsymbol{y},\boldsymbol{d}|\boldsymbol{z}}=\sum_{i=1}^{N}1\{\boldsymbol{Y}_{i}=\boldsymbol{y},\boldsymbol{D}_{i}=\boldsymbol{d},\boldsymbol{Z}_{i}=\boldsymbol{z}\}/\sum_{i=1}^{N}1\{\boldsymbol{Z}_{i}=\boldsymbol{z}\}$.

Figure \ref{fig:jtpa} reports the estimated partial ordering of welfare
$W_{\boldsymbol{\delta}}=E[Y_{2}(\boldsymbol{\delta}(\cdot))]$ (left)
and the resulting estimated set $\hat{\mathcal{D}}$ (right, highlighted
in red) that we estimate using $\{(\boldsymbol{Y}_{i},\boldsymbol{D}_{i},\boldsymbol{Z}_{i})\}_{i=1}^{9,223}$.
Although there exist welfares that cannot be ordered, we can conclude
\textit{with certainty} that allocating the program only to the low
earning type ($Y_{2}=0$) is welfare optimal, as it is the common
implication of Regimes 5 and 6 in $\hat{\mathcal{D}}$. Also, the
second best policy is to either allocate the program to the entire
population or none, while allocating it only to the high earning type
($Y_{2}=1$) produces the lowest welfare. This result is consistent
with the eligibility of Title II of the JTPA, which concerns individuals
with ``barriers to employment'' where the most common barriers are
unemployment spells and high-school dropout status (\citet{abadie2002instrumental}).
Possibly due to the fact that the first-stage instrument $Z_{1}$
is not strong enough, we have the two disconnected sub-DAGs and thus
the two elements in $\hat{\mathcal{D}}$, which are agnostic about
the optimal allocation in the first stage or the complementarity between
the first- and second- stage allocations.

\begin{figure}
\begin{centering}
\includegraphics[scale=0.5]{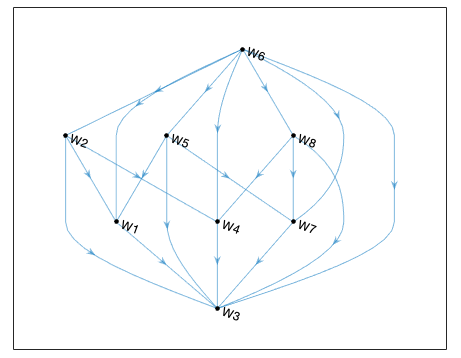}~{\footnotesize{}}%
\begin{tabular}[b]{|c|c|c|c|}
\hline 
{\footnotesize{}Regime \#} & {\footnotesize{}$\delta_{1}$} & {\footnotesize{}$\delta_{2}(1,\delta_{1})$} & {\footnotesize{}$\delta_{2}(0,\delta_{1})$}\tabularnewline
\hline 
\hline 
{\footnotesize{}1} & {\footnotesize{}0} & {\footnotesize{}0} & {\footnotesize{}0}\tabularnewline
\hline 
{\footnotesize{}2} & {\footnotesize{}1} & {\footnotesize{}0} & {\footnotesize{}0}\tabularnewline
\hline 
{\footnotesize{}3} & {\footnotesize{}0} & {\footnotesize{}1} & {\footnotesize{}0}\tabularnewline
\hline 
{\footnotesize{}4} & {\footnotesize{}1} & {\footnotesize{}1} & {\footnotesize{}0}\tabularnewline
\hline 
{\footnotesize{}5} & {\footnotesize{}0} & {\footnotesize{}0} & {\footnotesize{}1}\tabularnewline
\hline 
\textcolor{red}{\footnotesize{}6} & \textcolor{red}{\footnotesize{}1} & \textcolor{red}{\footnotesize{}0} & \textcolor{red}{\footnotesize{}1}\tabularnewline
\hline 
{\footnotesize{}7} & {\footnotesize{}0} & {\footnotesize{}1} & {\footnotesize{}1}\tabularnewline
\hline 
{\footnotesize{}8} & {\footnotesize{}1} & {\footnotesize{}1} & {\footnotesize{}1}\tabularnewline
\hline 
\multicolumn{1}{c}{} & \multicolumn{1}{c}{} & \multicolumn{1}{c}{} & \multicolumn{1}{c}{}\tabularnewline
\multicolumn{1}{c}{} & \multicolumn{1}{c}{} & \multicolumn{1}{c}{} & \multicolumn{1}{c}{}\tabularnewline
\end{tabular}{\footnotesize\par}
\par\end{centering}
\caption{Estimated Partial Ordering of $W_{\boldsymbol{\delta}}=E[Y_{1}(\delta_{1})]+E[Y_{2}(\boldsymbol{\delta}(\cdot))]$
and Estimated Set for $\boldsymbol{\delta}^{*}$ (red)}
\label{fig:jtpa-1}
\end{figure}
Figure \ref{fig:jtpa-1} reports the estimated partial ordering and
the estimated set with $W_{\boldsymbol{\delta}}=E[Y_{1}(\delta_{1})]+E[Y_{2}(\boldsymbol{\delta}(\cdot))]$.
Despite the partial ordering, $\hat{\mathcal{D}}$ is a singleton
for this welfare and $\boldsymbol{\delta}^{*}$ is estimated to be
Regime 6. According to this regime, the average lifetime earning is
maximized by allocating HS education to all individuals and the training
program to individuals with low pre-program earnings. As discussed
earlier, additional policy implications can be obtained by inspecting
suboptimal regimes. Interestingly, Regime 8, which allocates the treatments
regardless, is inferior to Regime 6. This can be useful knowledge
for policy makers especially because Regime 8 is the most ``expensive''
regime. Similarly, Regime 1, which does not allocate any treatments
regardless and thus is the least expensive regime, is superior to
Regime 3, which allocates the program to high-earning individuals.
The estimated partial ordering shows how more expensive policies
do not necessarily achieve greater welfare. Moreover, these conclusions
can be compelling as they are drawn without making arbitrary parametric
restrictions nor strong identifying assumptions.

Finally, as an alternative approach, we use $\{(\boldsymbol{Y}_{i},\boldsymbol{D}_{i},Z_{2i})\}_{i=1}^{9,223}$
for estimation, that is, we drop $Z_{1}$ and only use the exogenous
variation from $Z_{2}$. This reflects a possible concern that $Z_{1}$
may not be as valid as $Z_{2}$. Then, the estimated partial ordering
looks identical to the left panel of Figure \ref{fig:jtpa} whether
the targeted welfare is $E[Y_{2}(\boldsymbol{\delta}(\cdot))]$ or
$E[Y_{1}(\delta_{1})]+E[Y_{2}(\boldsymbol{\delta}(\cdot))]$. Clearly,
without $Z_{1}$, the procedure lacks the ability to determine the
first stage's best treatment. Note that, even though the partial
ordering for $E[Y_{2}(\boldsymbol{\delta}(\cdot))]$ is identical
for the case of one versus two instruments, the inference results
will reflect such difference by producing a larger confidence set
for the former case.

\section{Inference\label{sec:Discussions:-Estimation-and}}

Although we do not fully investigate inference in the current paper,
we briefly discuss it. To conduct inference on the optimal regime
$\boldsymbol{\delta}^{*}(\cdot)$, we can construct a confidence set
(CS) for $\mathcal{D}_{p}^{*}$ with the following procedure. We consider
a sequence of hypothesis tests, in which we eliminate regimes that
are (statistically) significantly inferior to others. This is a statistical
analog of the elimination procedure encoded in \eqref{eq:char_ID_set}
or \eqref{eq:char_ID_set2}. For each test given $\tilde{\mathcal{K}}\subset\mathcal{K}$,
we construct a null hypothesis that $W_{k}$ and $W_{k'}$ are not
comparable for all $k,k'\in\tilde{\mathcal{K}}$. Given \eqref{eq:LP},
the incomparability of $W_{k}$ and $W_{k'}$ is equivalent to $L_{k,k'}\le0\le U_{k,k'}$.
In constructing the null hypothesis, it is helpful to invoke strong
duality for the primal programs \eqref{eq:LP} and write the following
dual programs:
\begin{align}
U_{k,k'}= & \min_{\lambda}\tilde{p}'\lambda,\qquad s.t.\quad\tilde{B}'\lambda\ge\Delta_{k,k'}'\label{eq:LP_dual1}\\
L_{k,k'}= & \max_{\lambda}-\tilde{p}'\lambda,\qquad s.t.\quad\tilde{B}'\lambda\ge-\Delta_{k,k'}'\label{eq:LP_dual2}
\end{align}
where $\tilde{B}\equiv\left[\begin{array}{c}
B\\
\boldsymbol{1}'
\end{array}\right]$ is a $(d_{p}+1)\times d_{q}$ matrix with $\boldsymbol{1}$ being
a $d_{q}\times1$ vector of ones and $\tilde{p}\equiv\left[\begin{array}{c}
p\\
1
\end{array}\right]$ is a $(d_{p}+1)\times1$ vector. By using a vertex enumeration algorithm
(e.g., \citet{avis1992pivoting}), one can find all (or a relevant
subset) of vertices of the polyhedra $\{\lambda:\tilde{B}'\lambda\ge\Delta_{k,k'}'\}$
and $\{\lambda:\tilde{B}'\lambda\ge-\Delta_{k,k'}'\}$. Let $\Lambda_{U,k,k'}\equiv\{\lambda_{1},...,\lambda_{J_{U,k,k'}}\}$
and $\Lambda_{L,k,k'}\equiv\{\lambda_{1},...,\lambda_{J_{L,k,k'}}\}$
be the sets that collect such vertices, respectively. Then, it is
easy to see that $U_{k,k'}=\min_{\lambda\in\Lambda_{U,k,k'}}\tilde{p}'\lambda$
and $L_{k,k'}=\max_{\lambda\in\Lambda_{L,k,k'}}-\tilde{p}'\lambda$.
Therefore, the null hypothesis that $L_{k,k'}\le0\le U_{k,k'}$ can
be written as
\begin{align}
H_{0,\tilde{\mathcal{K}}} & :\tilde{p}'\lambda>0\text{ for all }\lambda\in\Lambda_{\tilde{\mathcal{K}}}.\label{eq:null}
\end{align}
where $\Lambda_{\tilde{\mathcal{K}}}\equiv\bigcup_{k,k'\in\tilde{\mathcal{K}}}\Lambda_{k,k'}$
with $\Lambda_{k,k'}\equiv\Lambda_{U,k,k'}\cup\Lambda_{L,k,k'}$.

Then, the procedure of constructing the CS, denoted as $\widehat{\mathcal{D}}_{CS}$,
is as follows: \textit{Step 0. Initially set $\tilde{\mathcal{K}}=\mathcal{K}$.
Step 1. Test $H_{0,\tilde{\mathcal{K}}}$ at level $\alpha$ with
test function $\phi_{\tilde{\mathcal{K}}}\in\{0,1\}$. Step 2. If
$H_{0,\tilde{\mathcal{K}}}$ is not rejected, define $\widehat{\mathcal{D}}_{CS}=\{\boldsymbol{\delta}_{k}(\cdot):k\in\tilde{\mathcal{K}}\}$;
otherwise eliminate vertex $k_{\tilde{\mathcal{K}}}$ from $\tilde{\mathcal{K}}$
and repeat from Step 1.} In Step 1, $T_{\tilde{\mathcal{K}}}\equiv\min_{k,k'\in\tilde{\mathcal{K}}}t_{k,k'}$
can be used as the test statistic for $H_{0,\tilde{\mathcal{K}}}$
where $t_{k,k'}\equiv\min_{\lambda\in\Lambda_{k,k'}}t_{\lambda}$
and $t_{\lambda}$ is a standard $t$-statistic. The distribution
of $T_{\tilde{\mathcal{K}}}$ can be estimated using bootstrap. In
Step 2, a candidate for $k_{\tilde{\mathcal{K}}}$ is $k_{\tilde{\mathcal{K}}}\equiv\arg\min_{k\in\tilde{\mathcal{K}}}\min_{k'\in\tilde{\mathcal{K}}}t_{k,k'}$. 

The eliminated vertices (i.e., regimes) are statistically suboptimal
regimes, which are already policy-relevant outputs of the procedure.
Note that the null hypothesis \eqref{eq:null} consists of multiple
inequalities. This incurs the issue of uniformity in that the null
distribution depends on binding inequalities, whose identities are
unknown. Such a problem has been studied in the literature, as in
\citet{hansen2005test}, \citet{andrews2010inference}, and \citet{chen2014testing}.
\citet{hansen2011model}'s bootstrap approach for constructing the
model confidence set builds on \citet{hansen2005test}. We apply a
similar inference method as in \citet{hansen2011model}, but in this
novel context and by being conscious about the computational challenge
of our problem. In particular, the dual problem \eqref{eq:LP_dual1}--\eqref{eq:LP_dual2}
and the vertex enumeration algorithm are introduced to ease the computational
burden in simulating the distribution of $T_{\tilde{\mathcal{K}}}$.
That is, the calculation of $\Lambda_{\tilde{\mathcal{K}}}$, the
computationally intensive step, occurs only once, and then for each
bootstrap sample, it suffices to calculate $\hat{p}$ instead of solving
the linear programs \eqref{eq:LP} for all $k,k'\in\tilde{\mathcal{K}}$.

Analogous to \citet{hansen2011model}, we can show that the resulting
CS has desirable properties. Let $H_{A,\tilde{\mathcal{K}}}$ be the
alternative hypothesis.

\begin{asCS}For any $\tilde{\mathcal{K}}$, (i) $\limsup_{n\rightarrow\infty}\Pr[\phi_{\tilde{\mathcal{K}}}=1|H_{0,\tilde{\mathcal{K}}}]\le\alpha$,
(ii) $\lim_{n\rightarrow\infty}\Pr[\phi_{\tilde{\mathcal{K}}}=1|H_{A,\tilde{\mathcal{K}}}]=1$,
and (iii) $\lim_{n\rightarrow\infty}\Pr[\boldsymbol{\delta}_{k_{\tilde{\mathcal{K}}}}(\cdot)\in\mathcal{\mathcal{D}}_{p}^{*}|H_{A,\tilde{\mathcal{K}}}]=0$.\end{asCS}

\begin{proposition}Under Assumption CS, it satisfies that $\lim\inf_{n\rightarrow\infty}\Pr[\mathcal{D}_{p}^{*}\subset\widehat{\mathcal{D}}_{CS}]\ge1-\alpha$
and $\lim_{n\rightarrow\infty}\Pr[\boldsymbol{\delta}(\cdot)\in\widehat{\mathcal{D}}_{CS}]=0$
for all $\boldsymbol{\delta}(\cdot)\notin\mathcal{\mathcal{D}}_{p}^{*}$.\end{proposition}

The procedure of constructing the CS does not suffer from the problem
of multiple testings. This is because the procedure stops as soon
as the first hypothesis is not rejected, and asymptotically, maximal
elements will not be questioned before all sub-optimal regimes are
eliminated. The resulting CS can also be used to conduct a specification
test for a less palatable assumption, such as Assumption M2. We can
refute the assumption when the CS under that assumption is empty.

Inference on the welfare bounds in \eqref{eq:LP-1} for a given regime
can be conducted by using recent results as in \citet{deb2017revealed},
who develop uniformly valid inference that can be applied to bounds
obtained via linear programming. Inference on optimized welfare $W_{\boldsymbol{\delta}^{*}}$
or $\max_{\boldsymbol{\delta}(\cdot)\in\widehat{\mathcal{D}}_{CS}}W_{\boldsymbol{\delta}}$
can also be an interesting problem. \citet{andrews2019inference}
consider inference on optimized welfare (evaluated at the estimated
policy) in the context of \citet{kitagawa2018should}, but with point-identified
welfare under the unconfoundedness assumption. Extending the framework
to the current setting with partially identified welfare and dynamic
regimes under treatment endogeneity would also be interesting future
work.

\appendix

\section{Appendix}

\subsection{Stochastic Regimes\label{subsec:Stochastic}}

For each $t=1,...,T$, define an adaptive \textit{stochastic} treatment
rule $\tilde{\delta}_{t}:\{0,1\}^{t-1}\times\{0,1\}^{t-1}\rightarrow[0,1]$
that allocates the probability of treatment: 
\begin{align}
\tilde{\delta}_{t}(\boldsymbol{y}^{t-1},\tilde{\boldsymbol{d}}^{t-1}) & =\tilde{d}_{t}\in[0,1].\label{eq:trt_rule-1}
\end{align}
Then, the vector of these $\tilde{\delta}_{t}$'s is a \textit{dynamic
stochastic regime} $\tilde{\boldsymbol{\delta}}(\cdot)\equiv\tilde{\boldsymbol{\delta}}^{T}(\cdot)\in\mathcal{D}_{stoch}$
where $\mathcal{D}_{stoch}$ is the set of all possible stochastic
regimes.\footnote{Dynamic stochastic regimes are considered in, e.g., \citet{murphy2001marginal},
\citet{murphy2003optimal}, and \citet{manski2004statistical}.} A deterministic regime is a special case where $\tilde{\delta}_{t}(\cdot)$
takes the extreme values of $1$ and $0$. Therefore, $\mathcal{D}\subset\mathcal{D}_{stoch}$
where $\mathcal{D}$ is the set of deterministic regimes. We define
$Y_{T}(\tilde{\boldsymbol{\delta}}(\cdot))$ with $\tilde{\boldsymbol{\delta}}(\cdot)\in\mathcal{D}_{stoch}$
as the counterfactual outcome $Y_{T}(\boldsymbol{\delta}(\cdot))$
where the deterministic rule $\delta_{t}(\cdot)=1$ is randomly assigned
with probability $\tilde{\delta}_{t}(\cdot)$ and $\delta_{t}(\cdot)=0$
otherwise for all $t\le T$. Finally, define
\begin{align*}
W_{\tilde{\boldsymbol{\delta}}} & \equiv\mathbb{E}[Y_{T}(\tilde{\boldsymbol{\delta}}(\cdot))],
\end{align*}
where $\mathbb{E}$ denotes an expectation over the counterfactual
outcome and the random mechanism defining a rule, and define $\tilde{\boldsymbol{\delta}}^{*}(\cdot)\equiv\arg\max_{\tilde{\boldsymbol{\delta}}(\cdot)\in\mathcal{D}_{stoch}}W_{\tilde{\boldsymbol{\delta}}}$.
The following theorem show that a deterministic regime is achieved
as being optimal even though stochastic regimes are allow.

\begin{theorem}\label{thm:stoch_regime}Suppose $W_{\tilde{\boldsymbol{\delta}}}\equiv\mathbb{E}[Y_{T}(\tilde{\boldsymbol{\delta}}(\cdot))]$
for $\tilde{\boldsymbol{\delta}}(\cdot)\in\mathcal{D}_{stoch}$ and
$W_{\boldsymbol{\delta}}\equiv E[Y_{T}(\boldsymbol{\delta}(\cdot))]$
for $\boldsymbol{\delta}(\cdot)\in\mathcal{D}$. It satisfies that
\begin{align*}
\boldsymbol{\delta}^{*}(\cdot)\equiv\arg\max_{\boldsymbol{\delta}(\cdot)\in\mathcal{D}}W_{\boldsymbol{\delta}} & =\arg\max_{\tilde{\boldsymbol{\delta}}(\cdot)\in\mathcal{D}_{stoch}}W_{\tilde{\boldsymbol{\delta}}}.
\end{align*}

\end{theorem}

By the law of iterative expectation, we have
\begin{align}
\mathbb{E}[Y_{T}(\tilde{\boldsymbol{\delta}}(\cdot))] & =\mathbb{E}\left[\mathbb{E}\left[\left.\cdots\mathbb{E}\left[\left.\mathbb{E}[Y_{T}(\tilde{\boldsymbol{d}})|\boldsymbol{Y}^{T-1}(\tilde{\boldsymbol{d}}^{T-1})]\right|\boldsymbol{Y}^{T-2}(\tilde{\boldsymbol{d}}^{T-2})\right]\cdots\right|Y_{1}(\tilde{d}_{1})\right]\right],\label{eq:mean_Y(delta)}
\end{align}
where the bridge variables $\tilde{\boldsymbol{d}}=(\tilde{d}_{1},...,\tilde{d}_{T})$
satisfy 
\begin{align*}
\tilde{d}_{1} & =\tilde{\delta}_{1},\\
\tilde{d}_{2} & =\tilde{\delta}_{2}(Y_{1}(\tilde{d}_{1}),\tilde{d}_{1}),\\
\tilde{d}_{3} & =\tilde{\delta}_{3}(\boldsymbol{Y}^{2}(\tilde{\boldsymbol{d}}^{2}),\tilde{\boldsymbol{d}}^{2}),\\
 & \vdots\\
\tilde{d}_{T} & =\tilde{\delta}_{T}(\boldsymbol{Y}^{T-1}(\tilde{\boldsymbol{d}}^{T-1}),\tilde{\boldsymbol{d}}^{T-1}).
\end{align*}
Given \eqref{eq:mean_Y(delta)}, we prove the theorem by showing that
the solution $\tilde{\boldsymbol{\delta}}^{*}(\cdot)$ can be justified
by backward induction in a finite-horizon dynamic programming. To
illustrate this with deterministic regimes when $T=2$, we have 
\begin{align}
\delta_{2}^{*}(y_{1},d_{1}) & =\arg\max_{d_{2}}E[Y_{2}(\boldsymbol{d})|Y_{1}(d_{1})=y_{1}],\label{eq:backward_T2_2}
\end{align}
and, by defining $V_{2}(y_{1},d_{1})\equiv\max_{d_{2}}E[Y_{2}(\boldsymbol{d})|Y_{1}(d_{1})=y_{1}]$,
\begin{align}
\delta_{1}^{*}= & \arg\max_{d_{1}}E[V_{2}(Y_{1}(d_{1}),d_{1})].\label{eq:backward_T2_1}
\end{align}
Then, $\boldsymbol{\delta}^{*}(\cdot)$ is equal to the collection
of these solutions: $\boldsymbol{\delta}^{*}(\cdot)=(\delta_{1}^{*},\delta_{2}^{*}(\cdot))$.

\begin{proof}First, given \eqref{eq:mean_Y(delta)}, the optimal
stochastic rule in the final period can be defined as
\begin{align*}
\tilde{\delta}_{T}^{*}(\boldsymbol{y}^{T-1},\tilde{\boldsymbol{d}}^{T-1})\equiv & \arg\max_{\tilde{d}_{T}\in[0,1]}\mathbb{E}[Y_{T}(\tilde{\boldsymbol{d}})|\boldsymbol{Y}^{T-1}(\tilde{\boldsymbol{d}}^{T-1})=\boldsymbol{y}^{T-1}].
\end{align*}
Define a value function at period $T$ as $V_{T}(\boldsymbol{y}^{T-1},\tilde{\boldsymbol{d}}^{T-1})\equiv\max_{\tilde{d}_{T}}\mathbb{E}[Y_{T}(\tilde{\boldsymbol{d}})|\boldsymbol{Y}^{T-1}(\tilde{\boldsymbol{d}}^{T-1})=\boldsymbol{y}^{T-1}]$.
Similarly, for each $t=1,...,T-1$, let 
\begin{align*}
\tilde{\delta}_{t}^{*}(\boldsymbol{y}^{t-1},\tilde{\boldsymbol{d}}^{t-1})\equiv & \arg\max_{\tilde{d}_{t}\in[0,1]}\mathbb{E}[V_{t+1}(\boldsymbol{Y}^{t}(\tilde{\boldsymbol{d}}^{t}),\tilde{\boldsymbol{d}}^{t})|\boldsymbol{Y}^{t-1}(\tilde{\boldsymbol{d}}^{t-1})=\boldsymbol{y}^{t-1}]
\end{align*}
and $V_{t}(\boldsymbol{y}^{t-1},\tilde{\boldsymbol{d}}^{t-1})\equiv\max_{\tilde{d}_{t}}\mathbb{E}[V_{t+1}(\boldsymbol{Y}^{t}(\tilde{\boldsymbol{d}}^{t}),\tilde{\boldsymbol{d}}^{t})|\boldsymbol{Y}^{t-1}(\tilde{\boldsymbol{d}}^{t-1})=\boldsymbol{y}^{t-1}]$.
Then, $\tilde{\boldsymbol{\delta}}^{*}(\cdot)=(\tilde{\delta}_{1}^{*},...,\tilde{\delta}_{T}^{*}(\cdot))$.
Since $\{0,1\}\subset[0,1]$, the same argument can apply for the
deterministic regime using the current framework but each maximization
domain being $\{0,1\}$. This analogously defines $\delta_{t}^{*}(\cdot)\in\{0,1\}$
for all $t$, and then $\boldsymbol{\delta}^{*}(\cdot)=(\delta_{1}^{*},...,\delta_{T}^{*}(\cdot))$,
similarly as in \citet{murphy2003optimal}.

Now, for the maximization problems above, let $\tilde{W}_{t}(\tilde{\boldsymbol{d}}^{t},\boldsymbol{y}^{t-1})$
represent the objective function at $t$ for $2\le t\le T$ with $\tilde{W}_{1}(\tilde{d}_{1})$
for $t=1$. By the definition of the stochastic regime, it satisfies
that
\begin{align*}
\tilde{W}_{t}(\tilde{\boldsymbol{d}}^{t},\boldsymbol{y}^{t-1}) & =\tilde{d}_{t}W_{t}(1,\tilde{\boldsymbol{d}}^{t-1},\boldsymbol{y}^{t-1})+(1-\tilde{d}_{t})W_{t}(0,\tilde{\boldsymbol{d}}^{t-1},\boldsymbol{y}^{t-1})\\
 & =\tilde{d}_{t}\left\{ W_{t}(1,\tilde{\boldsymbol{d}}^{t-1},\boldsymbol{y}^{t-1})-W_{t}(0,\tilde{\boldsymbol{d}}^{t-1},\boldsymbol{y}^{t-1})\right\} +W_{t}(0,\tilde{\boldsymbol{d}}^{t-1},\boldsymbol{y}^{t-1}).
\end{align*}
Therefore, $W_{t}(1,\tilde{\boldsymbol{d}}^{t-1},\boldsymbol{y}^{t-1})\ge W_{t}(0,\tilde{\boldsymbol{d}}^{t-1},\boldsymbol{y}^{t-1})$
or $1=\arg\max_{\tilde{d}_{t}\in\{0,1\}}\tilde{W}_{t}(\tilde{\boldsymbol{d}}^{t},\boldsymbol{y}^{t-1})$
if and only if $1=\arg\max_{\tilde{d}_{t}\in[0,1]}\tilde{W}_{t}(\tilde{\boldsymbol{d}}^{t},\boldsymbol{y}^{t-1})$.
Symmetrically, $0=\arg\max_{\tilde{d}_{t}\in\{0,1\}}\tilde{W}_{t}(\tilde{\boldsymbol{d}}^{t},\boldsymbol{y}^{t-1})$
if and only if $0=\arg\max_{\tilde{d}_{t}\in[0,1]}\tilde{W}_{t}(\tilde{\boldsymbol{d}}^{t},\boldsymbol{y}^{t-1})$.
This implies that $\tilde{\delta}_{t}^{*}(\cdot)=\delta_{t}^{*}(\cdot)$
for all $t=1,...,T$, which proves the theorem.\end{proof}

\subsection{Matrices in Section \ref{subsec:Data-Generating-Framework}\label{sec:Matrices}}

We show how to construct matrices $A_{k}$ and $B$ in \eqref{eq:LP_welfare}
and \eqref{eq:LP_constraint} for the linear programming \eqref{eq:LP}.
The construction of $A_{k}$ and $B$ uses the fact that any linear
functional of $\Pr[\boldsymbol{Y}(\boldsymbol{d})=\boldsymbol{y},\boldsymbol{D}(\boldsymbol{z})=\boldsymbol{d}]$
can be characterized as a linear combination of $q_{s}$. Although
the notation of this section can be somewhat heavy, if one is committed
to the use of linear programming instead of an analytic solution,
most of the derivation can be systematically reproduced in a standard
software, such as MATLAB and Python.

Consider $B$ first. By Assumption SX, we have 
\begin{align}
p_{\boldsymbol{y},\boldsymbol{d}|\boldsymbol{z}} & =\Pr[\boldsymbol{Y}(\boldsymbol{d})=\boldsymbol{y},\boldsymbol{D}(\boldsymbol{z})=\boldsymbol{d}]\nonumber \\
 & =\Pr[\boldsymbol{Y}(\boldsymbol{y}^{T-1},\boldsymbol{d})=\boldsymbol{y},\boldsymbol{D}(\boldsymbol{y}^{T-1},\boldsymbol{d}^{T-1},\boldsymbol{z})=\boldsymbol{d}]\nonumber \\
 & =\Pr[S:Y_{t}(\boldsymbol{y}^{t-1},\boldsymbol{d}^{t})=y_{t},D_{t}(\boldsymbol{y}^{t-1},\boldsymbol{d}^{t-1},\boldsymbol{z}^{t})=d_{t}\quad\forall t]\nonumber \\
 & =\sum_{s\in\mathcal{S}_{\boldsymbol{y},\boldsymbol{d}|\boldsymbol{z}}}q_{s},\label{eq:p_ydz}
\end{align}
where $\mathcal{S}_{\boldsymbol{y},\boldsymbol{d}|\boldsymbol{z}}\equiv\{S=\beta(\tilde{\boldsymbol{S}}):Y_{t}(\boldsymbol{y}^{t-1},\boldsymbol{d}^{t})=y_{t},D_{t}(\boldsymbol{y}^{t-1},\boldsymbol{d}^{t-1},\boldsymbol{z}^{t})=d_{t}\text{ }\forall t\}$,
$\tilde{\boldsymbol{S}}\equiv(\tilde{S}_{1},...,\tilde{S}_{T})$ with
$\tilde{S}_{t}\equiv(\{Y_{t}(\boldsymbol{y}^{t-1},\boldsymbol{d}^{t})\},\{D_{t}(\boldsymbol{y}^{t-1},\boldsymbol{d}^{t-1},\boldsymbol{z}^{t})\})$,
and $\beta(\cdot)$ is a one-to-one map that transforms a binary sequence
into a decimal value. Then, for a $1\times d_{q}$ vector $B_{\boldsymbol{y},\boldsymbol{d}|\boldsymbol{z}}$,
\begin{align*}
p_{\boldsymbol{y},\boldsymbol{d}|\boldsymbol{z}} & =\sum_{s\in\mathcal{S}_{\boldsymbol{y},\boldsymbol{d}|\boldsymbol{z}}}q_{s}=B_{\boldsymbol{y},\boldsymbol{d}|\boldsymbol{z}}q
\end{align*}
and the $d_{p}\times d_{q}$ matrix $B$ stacks $B_{\boldsymbol{y},\boldsymbol{d}|\boldsymbol{z}}$
so that $p=Bq$.

For $A_{k}$, recall $W_{\boldsymbol{\delta}_{k}}$ is a linear functional
of $q_{\boldsymbol{\delta}_{k}}(\boldsymbol{y})\equiv\Pr[\boldsymbol{Y}(\boldsymbol{\delta}_{k}(\cdot))=\boldsymbol{y}]$.
 For given $\boldsymbol{\delta}(\cdot)$, by repetitively applying
the law of iterated expectation, we can show
\begin{align}
 & \Pr[\boldsymbol{Y}(\boldsymbol{\delta}(\cdot))=\boldsymbol{y}]\nonumber \\
= & \Pr[Y_{T}(\boldsymbol{d})=y_{T}|\boldsymbol{Y}^{T-1}(\boldsymbol{d}^{T-1})=\boldsymbol{y}^{T-1}]\nonumber \\
 & \times\Pr[Y_{T-1}(\boldsymbol{d}^{T-1})=y_{T-1}|\boldsymbol{Y}^{T-2}(\boldsymbol{d}^{T-2})=\boldsymbol{y}^{T-2}]\times\cdots\times\Pr[Y_{1}(d_{1})=y_{1}],\label{eq:dist_Y(delta)2_1}
\end{align}
where, because of the appropriate conditioning in \eqref{eq:dist_Y(delta)2_1},
the bridge variables $\boldsymbol{d}=(d_{1},...,d_{T})$ satisfies
\begin{align*}
d_{1} & =\delta_{1},\\
d_{2} & =\delta_{2}(y_{1},d_{1}),\\
d_{3} & =\delta_{3}(\boldsymbol{y}^{2},\boldsymbol{d}^{2}),\\
 & \vdots\\
d_{T} & =\delta_{T}(\boldsymbol{y}^{T-1},\boldsymbol{d}^{T-1}).
\end{align*}
Therefore, \eqref{eq:dist_Y(delta)2_1} can be viewed as a linear
functional of $\Pr[\boldsymbol{Y}(\boldsymbol{d})=\boldsymbol{y}]$.
To illustrate, when $T=2$, the welfare defined as the average counterfactual
terminal outcome satisfies 
\begin{align}
E[Y_{T}(\boldsymbol{\delta}(\cdot))] & =\sum_{y_{1}}\Pr[Y_{2}(\delta_{1},\delta_{2}(Y_{1}(\delta_{1}),\delta_{1}))=1|Y_{1}(\delta_{1})=y_{1}]\Pr[Y_{1}(\delta_{1})=y_{1}]\nonumber \\
 & =\sum_{y_{1}}\Pr[Y_{2}(\delta_{1},\delta_{2}(y_{1},\delta_{1}))=1,Y_{1}(\delta_{1})=y_{1}].\label{eq:ex_T2}
\end{align}
Then, for a chosen $\boldsymbol{\delta}(\cdot)$, the values $\delta_{1}=d_{1}$
and $\delta_{2}(y_{1},\delta_{1})=d_{2}$ at which $Y_{2}(\delta_{1},\delta_{2}(y_{1},\delta_{1}))$
and $Y_{1}(\delta_{1})$ are defined is given in Table \ref{tab:regimes}
as shown in the main text. Therefore, $E[Y_{2}(\boldsymbol{\delta}(\cdot))]$
can be written as a linear functional of $\Pr[Y_{2}(d_{1},d_{2})=y_{2},Y_{1}(d_{1})=y_{1}]$.

Now, define a linear functional $h_{k}(\cdot)$ that (i) marginalizes
$\Pr[\boldsymbol{Y}(\boldsymbol{d})=\boldsymbol{y},\boldsymbol{D}(\boldsymbol{z})=\boldsymbol{d}]$
into $\Pr[\boldsymbol{Y}(\boldsymbol{d})=\boldsymbol{y}]$ and then
(ii) maps $\Pr[\boldsymbol{Y}(\boldsymbol{d})=\boldsymbol{y}]$ into
$\Pr[\boldsymbol{Y}(\boldsymbol{\delta}_{k}(\cdot))=\boldsymbol{y}]$
according to \eqref{eq:dist_Y(delta)2_1}. But recall that $\Pr[\boldsymbol{Y}(\boldsymbol{d})=\boldsymbol{y},\boldsymbol{D}(\boldsymbol{z})=\boldsymbol{d}]=\sum_{s\in\mathcal{S}_{\boldsymbol{y},\boldsymbol{d}|\boldsymbol{z}}}q_{s}$
by \eqref{eq:p_ydz}. Consequently, we have 
\begin{align*}
W_{k} & =f(q_{\boldsymbol{\delta}_{k}})=f(\Pr[\boldsymbol{Y}(\boldsymbol{\delta}_{k}(\cdot))=\cdot])\\
 & =f\circ h_{k}(\Pr[\boldsymbol{Y}(\cdot)=\cdot,\boldsymbol{D}(\boldsymbol{z})=\cdot]),\\
 & =f\circ h_{k}\left(\sum_{s\in\mathcal{S}_{\cdot,\cdot|\boldsymbol{z}}}q_{s}\right)\equiv A_{k}q.
\end{align*}
To continue the illustration \eqref{eq:ex_regime8} in the main text,
note that 
\begin{align*}
 & \Pr[\boldsymbol{Y}(1,1)=(1,1)]=\Pr[S:Y_{1}(1)=1,Y_{2}(1,1)=1]=\sum_{s\in\mathcal{S}_{11}}q_{s},
\end{align*}
where $\mathcal{S}_{11}\equiv\{S=\beta(\tilde{S}_{1},\tilde{S}_{2}):Y_{1}(1)=1,Y_{2}(1,1)=1\}$.
Similarly, we have 
\begin{align*}
 & \Pr[\boldsymbol{Y}(1,1)=(0,1)]=\Pr[S:Y_{1}(1)=0,Y_{2}(1,1)=1]=\sum_{s\in\mathcal{S}_{01}}q_{s},
\end{align*}
where $\mathcal{S}_{01}\equiv\{S=\beta(\tilde{S}_{1},\tilde{S}_{2}):Y_{1}(1)=0,Y_{2}(1,1)=1\}$.

\subsection{Proof of Theorem \ref{thm:DAG}}

Let $\mathcal{Q}_{p}\equiv\{q:Bq=p\}\cap\mathcal{Q}$ be the feasible
set. To prove part (i), first note that the sharp DAG can be explicitly
defined as $G(\mathcal{K},\mathcal{E}_{p})$ with
\begin{align*}
\mathcal{E}_{p} & \equiv\{(k,k')\in\mathcal{K}:A_{k}q>A_{k'}q\text{ for all }q\in\mathcal{Q}_{p}\}.
\end{align*}
Here, $A_{k}q>A_{k'}q$ for all $q\in\mathcal{Q}_{p}$ if and only
if $L_{k,k'}>0$ as $L_{k,k'}$ is the sharp lower bound of $(A_{k}-A_{k'})q$
in \eqref{eq:LP}. The latter is because the feasible set $\{q:Bq=p\text{ and }q\in\mathcal{Q}\}$
is convex and thus $\{\Delta_{k,k'}q:Bq=p\text{ and }q\in\mathcal{Q}\}$
is convex, which implies that any point between $[L_{k,k'},U_{k,k'}]$
is attainable.

To prove part (ii), it is helpful to note that $\mathcal{D}_{p}^{*}$
in \eqref{eq:ID_set} can be equivalently defined as
\begin{align*}
\mathcal{D}_{p}^{*} & \equiv\{\boldsymbol{\delta}_{k'}(\cdot):\nexists k\in\mathcal{K}\text{ such that }A_{k}q>A_{k'}q\text{ for all }q\in\mathcal{Q}_{p}\}\\
 & =\{\boldsymbol{\delta}_{k'}(\cdot):A_{k}q\le A_{k'}q\text{ for all }k\in\mathcal{K}\text{ and some }q\in\mathcal{Q}_{p}\}.
\end{align*}
Let $\tilde{\mathcal{D}}_{p}^{*}\equiv\{\boldsymbol{\delta}_{k'}(\cdot):\nexists k\in\mathcal{K}\text{ such that }L_{k,k'}>0\text{ and }k\neq k'\}$.
First, we prove that $\mathcal{D}_{p}^{*}\subset\tilde{\mathcal{D}}_{p}^{*}$.
Note that
\begin{align*}
\mathcal{D}\backslash\tilde{\mathcal{D}}_{p}^{*} & =\{\boldsymbol{\delta}_{k'}:L_{k,k'}>0\text{ for some }k\neq k'\}.
\end{align*}
Suppose $\boldsymbol{\delta}_{k'}\in\mathcal{D}\backslash\tilde{\mathcal{D}}_{p}^{*}$.
Then, for some $k\neq k'$, $(A_{k}-A_{k'})q\ge L_{k,k'}>0$ for all
$q\in\mathcal{Q}_{p}$. Therefore, for such $k$, $A_{k}q>A_{k'}q$
for all $q\in\mathcal{Q}_{p}$, and thus $\boldsymbol{\delta}_{k'}\notin\mathcal{D}_{p}^{*}\equiv\{\arg\max_{\boldsymbol{\delta}_{k}}A_{k}q:q\in\mathcal{Q}_{p}\}$.

Now, we prove that $\tilde{\mathcal{D}}_{p}^{*}\subset\mathcal{D}_{p}^{*}$.
Suppose $\boldsymbol{\delta}_{k'}\in\tilde{\mathcal{D}}_{p}^{*}$.
Then $\nexists k\neq k'$ such that $L_{k,k'}>0$. Equivalently, for
any given $k\neq k'$, either (a) $U_{k,k'}\le0$ or (b) $L_{k,k'}<0<U_{k,k'}$.
Consider (a), which is equivalent to $\max_{q\in\mathcal{Q}_{p}}(A_{k}-A_{k'})q\le0$.
This implies that $A_{k}q\le A_{k'}q$ for all $q\in\mathcal{Q}_{p}$.
Consider (b), which is equivalent to $\min_{q\in\mathcal{Q}_{p}}(A_{k}-A_{k'})q<0<\max_{q\in\mathcal{Q}_{p}}(A_{k}-A_{k'})q$.
This implies that $\exists q\in\mathcal{Q}_{p}$ such that $A_{k}q=A_{k'}q$.
Combining these implications of (a) and (b), it should be the case
that $\exists q\in\mathcal{Q}_{p}$ such that, for all $k\neq k'$,
$A_{k'}q\ge A_{k}q$. Therefore, $\boldsymbol{\delta}_{k}\in\mathcal{D}_{p}^{*}$.
$\boxempty$

\subsection{Alternative Characterization of the Identified Set\label{subsec:Characterizing-the-Identified}}

Given the DAG, the identified set of $\boldsymbol{\delta}^{*}(\cdot)$
can also be obtained as the collection of initial vertices of all
the directed paths of the DAG. For a DAG $G(\mathcal{K},\mathcal{E})$,
a \textit{directed path} is a subgraph $G(\mathcal{K}_{j},\mathcal{E}_{j})$
($1\le j\le J\le2^{\left|\mathcal{K}\right|}$) where $\mathcal{K}_{j}\subset\mathcal{K}$
is a totally ordered set with initial vertex $\tilde{k}_{j,1}$.\footnote{For example, in Figure \ref{fig:partial_order}(a), there are two
directed paths ($J=2$) with $V_{1}=\{1,2,3\}$ ($\tilde{k}_{1,1}=1$)
and $V_{2}=\{2,3,4\}$ ($\tilde{k}_{2,1}=4$).} In stating our main theorem, we make it explicit that the DAG calculated
by the linear programming is a function of the data distribution $p$.

\begin{theorem}\label{thm:ID_set}Suppose Assumptions SX and B hold.
Then, $\mathcal{D}_{p}^{*}$ defined in \eqref{eq:ID_set} satisfies
\begin{align}
\mathcal{D}_{p}^{*} & =\{\boldsymbol{\delta}_{\tilde{k}_{j,1}}(\cdot)\in\mathcal{D}:1\le j\le J\},\label{eq:obtain_D*}
\end{align}
where $\tilde{k}_{j,1}$ is the initial vertex of the directed path
$G(\mathcal{K}_{p,j},\mathcal{E}_{p,j})$ of $G(\mathcal{K},\mathcal{E}_{p})$.\end{theorem}

\begin{proof}Let $\tilde{\mathcal{D}}^{*}\equiv\{\boldsymbol{\delta}_{\tilde{k}_{j,1}}(\cdot)\in\mathcal{D}:1\le j\le J\}$.
First, note that since $\tilde{k}_{j,1}$ is the initial vertex of
directed path $j$, it should be that $W_{\tilde{k}_{j,1}}\ge W_{\tilde{k}_{j,m}}$
for any $\tilde{k}_{j,m}$ in that path by definition. We begin by
supposing $\mathcal{D}_{p}^{*}\supset\tilde{\mathcal{D}}^{*}$. Then,
there exist $\boldsymbol{\delta}^{*}(\cdot;q)=\arg\max_{\boldsymbol{\delta}_{k}(\cdot)\in\mathcal{D}}A_{k}q$
for some $q$ that satisfies $Bq=p$ and $q\in\mathcal{Q}$, but which
is not the initial vertex of any directed path. Such $\boldsymbol{\delta}^{*}(\cdot;q)$
cannot be other (non-initial) vertices of any paths as it is contradiction
by the definition of $\boldsymbol{\delta}^{*}(\cdot;q)$. But the
union of all directed paths is equal to the original DAG, therefore
there cannot exist such $\boldsymbol{\delta}^{*}(\cdot;q)$.

Now suppose $\mathcal{D}_{p}^{*}\subset\tilde{\mathcal{D}}^{*}$.
Then, there exists $\boldsymbol{\delta}_{\tilde{k}_{j,1}}(\cdot)\neq\boldsymbol{\delta}^{*}(\cdot;q)=\arg\max_{\boldsymbol{\delta}_{k}(\cdot)\in\mathcal{D}}A_{k}q$
for some $q$ that satisfies $Bq=p$ and $q\in\mathcal{Q}$. This
implies that $W_{\tilde{k}_{j,1}}<W_{\tilde{k}}$ for some $\tilde{k}$.
But $\tilde{k}$ should be a vertex of the same directed path (because
$W_{\tilde{k}_{j,1}}$ and $W_{\tilde{k}}$ are ordered), but then
it is contradiction as $\tilde{k}_{j,1}$ is the initial vertex. Therefore,
$\mathcal{D}_{p}^{*}=\tilde{\mathcal{D}}^{*}$.\end{proof}

\subsection{Proof of Theorem \ref{thm:topo_sort_ID_set}}

Given Theorem \ref{thm:ID_set}, proving $\tilde{\mathcal{D}}^{*}=\{\boldsymbol{\delta}_{k_{l,1}}(\cdot):1\le l\le L_{G}\}$
will suffice. Recall $\tilde{\mathcal{D}}^{*}\equiv\{\boldsymbol{\delta}_{\tilde{k}_{j,1}}(\cdot)\in\mathcal{D}:1\le j\le J\}$
where $\tilde{k}_{j,1}$ is the initial vertex of the directed path
$G(\mathcal{K}_{p,j},\mathcal{E}_{p,j})$. When all topological sorts
are singletons, the proof is trivial so we rule out this possibility.
Suppose $\tilde{\mathcal{D}}^{*}\supset\{\boldsymbol{\delta}_{k_{l,1}}(\cdot):1\le l\le L_{G}\}$.
Then, for some $l$, there should exist $\boldsymbol{\delta}_{k_{l,m}}(\cdot)$
for some $m\neq1$ that is contained in $\tilde{\mathcal{D}}^{*}$
but not in $\{\boldsymbol{\delta}_{k_{l,1}}(\cdot):1\le l\le L_{G}\}$,
i.e., that satisfies either (i) $W_{k_{l,1}}>W_{k_{l,m}}$ or (ii)
$W_{k_{l,1}}$ and $W_{k_{l,m}}$ are incomparable and thus either
$W_{k_{l',1}}>W_{k_{l,m}}$ for some $l'\neq l$ or $W_{k_{l,m}}$
is a singleton in another topological sort. Consider case (i). If
$\boldsymbol{\delta}_{k_{l,1}}(\cdot)\in\mathcal{D}_{j}$ for some
$j$, then it should be that $\boldsymbol{\delta}_{k_{l,m}}(\cdot)\in\mathcal{D}_{j}$
as $\boldsymbol{\delta}_{k_{l,1}}(\cdot)$ and $\boldsymbol{\delta}_{k_{l,m}}(\cdot)$
are comparable in terms of welfare, but then $\boldsymbol{\delta}_{k_{l,m}}(\cdot)\in\tilde{\mathcal{D}}^{*}$
contradicts the fact that $\boldsymbol{\delta}_{k_{l,1}}(\cdot)$
the initial vertex of the topological sort. Consider case (ii). The
singleton case is trivially rejected since if the topological sort
a singleton, then $\boldsymbol{\delta}_{k_{l,m}}(\cdot)$ should have
been already in $\{\boldsymbol{\delta}_{k_{l,1}}(\cdot):1\le l\le L_{G}\}$.
In the other case, since the two welfares are not comparable, it should
be that $\boldsymbol{\delta}_{k_{l,m}}(\cdot)\in\mathcal{D}_{j'}$
for $j'\neq j$. But $\boldsymbol{\delta}_{k_{l,m}}(\cdot)$ cannot
be the one that delivers the largest welfare since $W_{k_{l',1}}>W_{k_{l,m}}$
where $\boldsymbol{\delta}_{k_{l',1}}(\cdot)$. Therefore $\boldsymbol{\delta}_{k_{l,m}}(\cdot)\in\tilde{\mathcal{D}}^{*}$
is contradiction. Therefore there is no element in $\tilde{\mathcal{D}}^{*}$
that is not in $\{\boldsymbol{\delta}_{k_{l,1}}(\cdot):1\le l\le L_{G}\}$.

Now suppose $\tilde{\mathcal{D}}^{*}\subset\{\boldsymbol{\delta}_{k_{l,1}}(\cdot):1\le l\le L_{G}\}$.
Then for $l$ such that $\boldsymbol{\delta}_{k_{l,1}}(\cdot)\notin\tilde{\mathcal{D}}^{*}$,
either $W_{k_{l,1}}$ is a singleton or $W_{k_{l,1}}$ is an element
in a non-singleton topological sort. But if it is a singleton, then
it is trivially totally ordered and is the maximum welfare, and thus
$\boldsymbol{\delta}_{k_{l,1}}(\cdot)\notin\tilde{\mathcal{D}}^{*}$
is contradiction. In the other case, if $W_{k_{l,1}}$ is a maximum
welfare, then $\boldsymbol{\delta}_{k_{l,1}}(\cdot)\notin\tilde{\mathcal{D}}^{*}$
is contradiction. If it is not a maximum welfare, then it should be
a maximum in another topological sort, which is contradiction in either
case of being contained in $\{\boldsymbol{\delta}_{k_{l,1}}(\cdot):1\le l\le L_{G}\}$
or not. This concludes the proof that $\tilde{\mathcal{D}}^{*}=\{\boldsymbol{\delta}_{k_{l,1}}(\cdot):1\le l\le L_{G}\}$.
$\boxempty$

\subsection{Proof of Lemma \ref{lem:vyt_1}}

Conditional on $(\boldsymbol{Y}^{t-1},\boldsymbol{D}^{t-1},\boldsymbol{Z}^{t-1})=(\boldsymbol{y}^{t-1},\boldsymbol{d}^{t-1},\boldsymbol{z}^{t-1})$,
it is easy to show that \eqref{eq:model2_only} implies Assumption
M1. Suppose $\pi_{t}(\boldsymbol{y}^{t-1},\boldsymbol{d}^{t-1},\boldsymbol{z}^{t-1},1)>\pi_{t}(\boldsymbol{y}^{t-1},\boldsymbol{d}^{t-1},\boldsymbol{z}^{t-1},1)$
as $\pi_{t}(\cdot)$ is a nontrivial function of $Z_{t}$. Then, we
have 
\begin{align*}
1\{\pi_{t}(\boldsymbol{y}^{t-1},\boldsymbol{d}^{t-1},\boldsymbol{z}^{t-1},1)\ge V_{t}\} & \ge1\{\pi_{t}(\boldsymbol{y}^{t-1},\boldsymbol{d}^{t-1},\boldsymbol{z}^{t-1},0)\ge V_{t}\}
\end{align*}
w.p.1, or equivalently, $D_{t}(\boldsymbol{z}^{t-1},1)\ge D_{t}(\boldsymbol{z}^{t-1},0)$
w.p.1. Suppose $\pi_{t}(\boldsymbol{y}^{t-1},\boldsymbol{d}^{t-1},\boldsymbol{z}^{t-1},1)<\pi_{t}(\boldsymbol{y}^{t-1},\boldsymbol{d}^{t-1},\boldsymbol{z}^{t-1},1)$.
Then, by a parallel argument, $D_{t}(\boldsymbol{z}^{t-1},1)\le D_{t}(\boldsymbol{z}^{t-1},0)$
w.p.1.

Now, we show that Assumption M1 implies \eqref{eq:model2_only} conditional
on $(\boldsymbol{Y}^{t-1},\boldsymbol{D}^{t-1},\boldsymbol{Z}^{t-1})$.
For each $t$, Assumption SX implies $Y_{t}(\boldsymbol{d}^{t}),D_{t}(\boldsymbol{z}^{t})\perp\boldsymbol{Z}^{t}|(\boldsymbol{Y}^{t-1}(\boldsymbol{d}^{t-1}),\boldsymbol{D}^{t-1}(\boldsymbol{z}^{t-1}),\boldsymbol{Z}^{t-1})$,
which in turn implies the following conditional independence: 
\begin{align}
Y_{t}(\boldsymbol{d}^{t}),D_{t}(\boldsymbol{z}^{t}) & \perp\boldsymbol{Z}^{t}|(\boldsymbol{Y}^{t-1},\boldsymbol{D}^{t-1},\boldsymbol{Z}^{t-1}).\label{eq:SX_2}
\end{align}
Conditional on $(\boldsymbol{Y}^{t-1},\boldsymbol{D}^{t-1},\boldsymbol{Z}^{t-1})$,
\eqref{eq:model2_only} and \eqref{eq:SX_2} correspond to Assumption
S-1 in \citet{vytlacil2002independence}. Assumption R(i) and \eqref{eq:SX_2}
correspond to Assumption L-1, and Assumption M1 corresponds to Assumption
L-2 in \citet{vytlacil2002independence}. Therefore, the desired result
follows by Theorem 1 of \citet{vytlacil2002independence}. $\boxempty$

\subsection{Proof of Lemma \ref{lem:vyt_2}}

We are remained to prove that, conditional on $(\boldsymbol{Y}^{t-1},\boldsymbol{D}^{t-1})$,
\eqref{eq:model1} is equivalent to the second part of Assumption
M2. But this proof is analogous to the proof of Lemma \ref{lem:vyt_1}
by replacing the roles of $D_{t}$ and $Z_{t}$ with those of $Y_{t}$
and $D_{t}$, respectively. Therefore, we have the desired result.
$\boxempty$

\bibliographystyle{ecta}
\bibliography{optDTR}

\end{document}